\documentclass[journal,twoside,web]{ieeecolor}
\usepackage{generic}
\usepackage{cite}
\usepackage{amsmath,amssymb,amsfonts}
\usepackage{algorithmic}
\usepackage{graphicx}
\usepackage{algorithm,algorithmic}
\usepackage{hyperref}
\hypersetup{hidelinks=true}
\usepackage{textcomp}
\def\BibTeX{{\rm B\kern-.05em{\sc i\kern-.025em b}\kern-.08em
		T\kern-.1667em\lower.7ex\hbox{E}\kern-.125emX}}
\newtheorem{theorem}{Theorem}
\newtheorem{definition}{\textbf{Definition}}
\newtheorem{lemma}{Lemma}
\newtheorem{remark}{Remark}
\newtheorem{corollary}{Corollary}

\begin{document}
	\title{Observer-based Differentially Private Consensus  for Linear Multi-agent Systems}
	\author{Xiaofeng Zong, \IEEEmembership{Member, IEEE}, Ming-Yu Wang, Jimin Wang, \IEEEmembership{Member, IEEE}, and Ji-Feng Zhang, \IEEEmembership{Fellow, IEEE}
		\thanks{The research was supported by the National Natural Science Foundations of China under Grants 62522319, 62433020, 62473347, 62261136550, and T2293770. Corresponding author: Ji-Feng Zhang.}
		\thanks{Xiaofeng Zong is with the School of Future Technology, China University of Geosciences, Wuhan 430074, China, the School of Automation, China University of Geosciences, Wuhan 430074, China, and also with the Hubei Key Laboratory of Advanced Control and Intelligent Automation for Complex Systems, Wuhan 430074, China (e-mail: zongxf@cug.edu.cn)}
		\thanks{Ming-Yu Wang is with the School of Future Technology, China University of Geosciences, Wuhan 430074, China (e-mail: wangMing-Yu@cug.edu.cn)}
		\thanks{Jimin Wang is with the School of Automation and Electrical Engineering, University of Science and Technology Beijing, Beijing 100083, and also with the Key Laboratory of Knowledge Automation for Industrial Processes, Ministry of Education, Beijing 100083, China (e-mail:jimwang@ustb.edu.cn)}
		\thanks{Ji-Feng Zhang is with the School of Automation and Electrical Engineering, Zhongyuan University of Technology, Zheng Zhou 450007; and also with the State Key Laboratory of Mathematical Sciences, Academy of Mathematics and Systems Science, Chinese Academy of Sciences, Beijing 100190, China. (e-mail: jif@iss.ac.cn)}}

	\maketitle
	
	\begin{abstract}
		This paper investigates the differentially private consensus problem for general linear multi-agent systems (MASs) based on output feedback protocols.
		To protect the output information, which is considered private data and may be at high risk of exposure, Laplace noise is added to the information exchange.
		The conditions for achieving mean square and almost sure consensus in observer-based MASs are established using the backstepping method and the convergence theory for nonnegative almost supermartingales. It is shown that the separation principle remains valid for the consensus problem of linear MASs with decaying Laplace noise.  Furthermore, the convergence rate is provided.
		Then, a joint design framework is developed for state estimation gain, feedback control gain, and noise to ensure the preservation of $\epsilon$-differential privacy. The output information of each agent is shown to be protected at every time step.
		Finally, sufficient conditions are established for simultaneously achieving consensus and preserving differential privacy for linear MASs utilizing both full-order and reduced-order observers. Meanwhile, an $\epsilon^\star$-differentially private consensus is achieved to meet the desired privacy level. Two simulation examples are provided to validate the theoretical results.
	\end{abstract}
	
	\section{Introduction}
	Recent decades have witnessed the great progress of consensus theory for multi-agent systems (MASs). Various kinds of consensus are studied, such as average consensus \cite{ac1,ac2}, bipartite consensus\cite{bc1,bc2}, group consensus\cite{li2022group,li2022group1}, leader-follower consensus\cite{lc1,lc2}, and finite-time consensus\cite{fc1,fc2}. In most consensus protocols, agents achieve consensus by exchanging their state information with their neighbors. However, in many scenarios, agents do not want their real state information to be exposed to other agents. For example, in a multi-unmanned aerial vehicles (UAVs) system, the trajectory of each UAV is private and does not want to be obtained by other UAVs.  {Another example can be found in smart grids, where multiple grid-forming inverters use consensus algorithms to achieve reactive power sharing. Each distributed generators needs to obtain the system's average reactive power value. Although traditional dynamic consensus observers can accurately track the average value, their information exchange process allows neighboring nodes to reverse-engineer the original local states from the received state data. Eavesdropping agents can precisely deduce the true local state values simply by accessing the observed states of neighboring nodes and leveraging algorithmic relationships, which severely compromises the commercial confidentiality of distributed generator operations \cite{Du2021}. Similar privacy concerns also arise in autonomous vehicle 
		coordination.} Thus, the privacy-preserving problem must be considered simultaneously.
	
	{To solve this problem, many methods have been proposed, such as encryption mechanisms \cite{Tan2023,Km}, partial information transmission \cite{pit1,pit2}, C-R lower bound \cite{Guo2025}, and differential privacy\cite{CD,20}.} In recent years, differential privacy has attracted much attention because of its robustness to various attacks.  {The core of differential privacy is a rigorous randomization mechanism: when querying a database, it adds controllable noise to the results based on predefined parameters. This ensures that the presence or absence of any single individual's data does not significantly affect the final output. As a result, even if attackers obtain the processed results, they cannot infer specific individuals' sensitive information through differences in the data. However, differential privacy also has inherent limitations: the injection of noise, while providing privacy guarantees, inevitably affects the accuracy of average consensus and may slow down the system convergence rate (\cite{20,95,22,90,91,res}).}
	
	Moreover, in practical applications, directly measuring the state information of a system is often challenging due to technical limitations and economic constraints on measurement equipment. Consequently, it is generally infeasible to obtain all the state information directly. To address this, system states can be reconstructed by designing an observer. In traditional observer-based MASs, the output variables used to construct the observer may contain sensitive information, such as the position or velocity of the agents, making the protection of these output variables critically important. In \cite{l1}, McGlinchey et al. proposed a method to protect the output of a free-motion system using a full-order observer. However, this approach cannot be directly applied to achieve differentially private consensus in MASs, as the corresponding control input alters the structure of the closed-loop system, and its stability conditions have not been well established.  {Meanwhile, the expansion from full-order observer to reduced-order is also  very important. This can bring about significant improvements in computational efficiency, reductions in system complexity, and savings in engineering costs. Therefore, in the design of the vast majority of actual control systems, as long as conditions permit, engineers will give priority to designing reduced-dimensional observers.} This work aims to bridge these gaps by proposing a novel approach that addresses these challenges effectively.
	
	In this article, we tackle the observer-based differentially private consensus problem for general linear MASs. To protect the output information, Laplace noise is incorporated into the information exchange process. In addition to developing a differentially private consensus framework based on the full-order observer, we also extend our approach to include a differentially private consensus mechanism utilizing a reduced-order observer. Furthermore, a novel feedback stabilization theory based on noisy state observers-encompassing both full-order and reduced-order observers is proposed. This advancement addresses a critical gap in the stochastic stabilization of linear stochastic systems using state observers.
	
	\subsection{State of the Art}
	The concept of differential privacy was first proposed in \cite{CD}, and later extended to the field of MASs to preserve the state information of agents. Recent advances in differential privacy research for MASs have primarily concentrated on first-order dynamics while ensuring the protection of initial state information. In \cite{20}, sufficient conditions were established for achieving average consensus in multi-agent systems while preserving privacy through the injection of exponentially decaying additive Laplace noise. The work provided a thorough exploration of the trade-off between accuracy and privacy.
	In \cite{Mo2016}, by adding Gaussian noise, the initial state was protected while achieving exact average consensus, and it was proven that the proposed algorithm attained minimal privacy leakage.
	In \cite{22}, asymptotic average consensus was achieved under $\epsilon$-differential privacy. The study not only derived the mean square convergence rate but also revealed that slower noise decay rates could degrade the convergence rate. Furthermore, an optimal parameter configuration was derived to minimize the variance of the convergence point. In \cite{95}, Wang et al. introduced an event-triggered scheme that ensured consensus while maintaining differential privacy, and further investigated parameter optimization strategies under a fixed privacy level $\epsilon$ to maximize system accuracy. In \cite{90}, Gao et al. gave some sufficient conditions in terms of local information to achieve average consensus and preserve differential privacy simultaneously by using an event-triggered control. For the case with quantized communication, some sufficient conditions to achieve differentially private consensus were given in \cite{91}. In \cite{res}, a differentially private mean-subsequence-reduced algorithm was introduced, enabling resilient consensus in multi-agent systems over directed networks while preserving differential privacy. It ensured non-faulty agents converged despite faulty nodes, provided the network was sufficiently robust, and protected initial states via noise injection. {In \cite{yang2024}, a differentially private average consensus algorithm was introduced in which a dynamic self-triggered mechanism was incorporated, thereby safeguarding initial-state privacy while maintaining execution efficiency. In \cite{zhou2025}, a probabilistic communication strategy was proposed for differentially private consensus under dynamic topologies. Convergence conditions were established, and a convex optimization method was developed to accelerate convergence under communication constraints. In \cite{tian2025}, a differentially private absolute weighted mean subsequence reduction algorithm was proposed for MASs over signed digraphs with misbehaving agents. It guaranteed resilient bipartite consensus while preserving differential privacy of initial conditions, maintained accuracy under sufficient network connectivity, and tolerated up to $f$ faulty agents.}
	
	Meanwhile, the privacy-preserving problem of second-order systems and high-order MASs has received growing attention. In \cite{99}, for the case of second-order MASs with quantized communication, some sufficient conditions were given for MASs with connected graph to achieve both consensus and privacy preservation. In \cite{erfen}, a hierarchical mechanism was introduced for general linear MASs, which employs a classifier and noise injectors to preserve private information while achieving mean square consensus.   {To provide a more systematic comparison with related works,  Table 1 is constructed.}
	
	In summary, the aforementioned literature shares the following common characteristics: (i)  {The proposed mechanisms can only protect the system's initial state, which exposes it to serious security risks. On the one hand, the exposure of process information enables attackers to infer sensitive information such as network topology and control strategies. On the other hand, potentially ``honest-but-curious'' agents within the system may continuously collect states from their neighbors, leading to gradual privacy leakage; }(ii) To achieve consensus and differential privacy, all these studies add exponentially decaying Laplace noise into the system, which inevitably leads to a rapid degradation of system state preservation over time. Therefore, we raise the following fundamental questions: Can a differential privacy framework be developed to ensure whole-process privacy preservation for general linear systems? Meanwhile, is it possible to relax the exponential decay requirement for noise, allowing a slower decay rate while still preserving privacy? How can consensus and privacy preservation be achieved simultaneously? These considerations directly motivated the research presented in this work.

	\begin{table*}[t]
		\centering
		\caption{ {A comparison of some related references.}}
		\centering
		{ 
			\begin{tabular}{l l l l p{2.7cm} l p{3.5cm}}
				\hline
				\textbf{Reference} & \textbf{Agent model} & \textbf{Private information} & \textbf{State known?} & \textbf{Added noise} & \textbf{Privacy criterion} & \textbf{Consensus performance} \\
				\hline
				\cite{20} & first-order & initial states & known & exponentially decaying Laplace noise &  $\epsilon$-DP & iterative synchronous consensus \\			
				\cite{Mo2016} & first-order & initial states & known & exponentially decaying Gaussian noise &  \( {e}_{i} \notin  \mathbb{D} \) & exact average consensus in the mean square sense\\			
				\cite{22,zhou2025} & first-order & initial states & known & exponentially decaying Laplace noise & $\epsilon$-DP & average consensus in the mean square sense and almost sure sense \\			
				\cite{95,yang2024} & first-order & initial states & known & exponentially decaying Laplace noise & $\epsilon$-DP & average consensus in the mean square sense  \\			
				\cite{90,91,res} & first-order & initial states & known & exponentially decaying Laplace noise & $\epsilon$-DP  & consensus in the mean square sense  \\
				\cite{l1} & general linear & entire output trajectories & unknown & constant Laplace noise & $\epsilon$-DP & $\times$ \\			
				\cite{tian2025} & first-order & initial states & known & exponentially decaying Laplace noise & $\epsilon$-DP & resilient bipartite consensus in the mean square sense \\			
				\cite{99} & second-order & initial states & known & exponentially decaying Laplace noise & $\epsilon$-DP & dynamic average consensus in the mean square sense \\			
				\cite{erfen} & general linear & initial states & known & exponentially decaying Laplace noise & $\epsilon$-DP & bipartite consensus in the mean square sense \\			
				\textbf{Our work} & general linear & entire output trajectories & unknown & decaying Laplace noise & $\epsilon$-DP \& $\epsilon^\star$-DP & dynamic average consensus in the mean square sense and almost sure sense \\
				\hline
			\end{tabular}
		} 
	\end{table*}

	\subsection{Contributions}
	The main contributions of this work are outlined as follows.
	
	\begin{itemize}
		\item  The conditions for achieving mean square and almost sure consensus in observer-based MASs are established using the backstepping method and the convergence theory for nonnegative almost supermartingales. Meanwhile, the mean square convergence rate of the scale parameter under exponential decay is obtained. Furthermore, it is demonstrated that the separation principle remains valid for the consensus problem of linear MASs with decaying Laplace noise.
		\item  {A unified design framework is proposed for jointly tuning the state-estimation gain, feedback control gain, and injected noise to guarantee $\epsilon$-differential privacy at every time step. Meanwhile, a generalized adjacency relation is introduced to allow broader deviations between adjacent datasets, and the requirement that the noise variance decay exponentially is relaxed.	}
		\item  Sufficient conditions for simultaneously achieving consensus and preserving differential privacy are established for linear MASs utilizing both full-order and reduced-order observers. This result distinguishes our work from He et al. \cite{81}, who demonstrated that achieving average consensus and differential privacy simultaneously is impossible for some first-order systems.
		\item { Parameter design conditions are derived for a prescribed privacy level. For any given privacy budget $\epsilon^\star$, explicit inequalities are provided for jointly selecting the observer gain $L$, feedback gain $K$, and noise parameters, so that the system achieves mean-square and almost-sure consensus while satisfying $\epsilon^\star$-differential privacy.}
	\end{itemize}

	\subsection{Outline and Notations}
	This work is organized as follows. In Section \ref{s2}, we present the basic graph theory and the problem formulation of this article.  {Section \ref{Sfull} provides  the sufficient conditions for consensus, $\epsilon$-differential privacy, differentially private consensus, and $\epsilon^\star$-differentially private consensus based on the full-order observer.} A number of interesting corollaries are presented, serving as significant extensions of the theorem. {Moreover, the convergence rate is given.}  These results are then extended to MASs based on the reduced-order observer in Section \ref{Sreduce}. To validate our theoretical findings, some simulation results are given in Section \ref{S5}. Finally, Section \ref{S6} concludes this article.
	
	Notations:  $\mathbb{R}^{n}$ denotes $n$-dimensional Euclidean space. $(\mathbb{R}^{n})^\mathbb{N}$ denotes the space of vector-valued sequences in $\mathbb{R}^{n}$. $\mathbb{Z}_{\geq 0}$ denotes the set of nonnegative integers. $M^{\mathrm{T}}$ denotes the transpose of matrix $M$ and  $M^{-1}$ denotes the inverse of $M$. $I_q$ denotes a $q$-dimensional identity matrix. $\|X\|_1$, $\|X\|$, and $\|X\|_F$ denote the $1$-norm, $2$-norm, and Frobenius norm of $X$, respectively.   $\mathbb{P}\{Y\}$ denotes the probability of  $Y$. $E[Y] $ denotes the expectation  of  $Y$. For a square matrix, the spectral radius is denoted as $ \rho (\cdot)$. $ P \otimes Q$ denotes the Kronecker product of matrices $P$ and $Q$.
	$\gamma \sim \operatorname{Lap}(0, b)$  stands for that $\gamma\in \mathbb{R}$ satisfies the Laplace distribution with a mean of 0 and a scale parameter of $b>0$.
	Denote $\eta \sim \operatorname{Lap}(b)^n$ for random variable $\eta \in \mathbb{R}^{n}$ with each of its components $\eta_j \sim \operatorname{Lap}(0, b)$ being independent and identically distributed, and we have $\mathbb{P}(\eta ; b)=\left(\frac{1}{2 b}\right)^n \exp \left(-\frac{\|\eta \|_1}{b}\right)$ (see \cite{LeNy2014} for details). We write $f(k) = \mathcal{O}(g(k))$ if there exist positive constants \(M\) and \(k_0\) such that for all \(k \ge k_0\), $|f(k)| \le M \cdot |g(k)|$.  { \(\mathcal{B}(X)\) denotes the Borel  $\sigma$-algebra on $X$, generated by all open sets.}

	\section{Preliminaries and Problem Formulation}\label{s2}
	
	In this paper, the connected undirected graph $\mathcal{G}=(\mathcal{V}, \mathcal{A})$ is considered, where $\mathcal{V}=\{1,2,\cdots,N\}$ represents the set of $N$ agents and $\mathcal{A} \subseteq \mathcal{V} \times \mathcal{V}$ denotes the set of edges. The network structure is characterized by the adjacency matrix $A_G=\left[a_{i j}\right] \in \mathbb{R}^{N \times N}$, whose elements $a_{ij}$ satisfy $a_{ij} = a_{ji}$ due to the undirected nature of the graph, with $a_{ij} = 1$ if agents $i$ and $j$ are connected and $a_{ij} = 0$ otherwise.  Denote $d_i=\sum_{j=1}^{N} a_{i j}$ as the degree of $i$ and $\mathcal{D}=\operatorname{diag}\left\{t_{1}, \cdots, t_{N}\right\}$ as the degree matrix. The Laplacian matrix is  $L_G=\mathcal{D}-A_G$, whose eigenvalues are denoted as $0=\lambda_{1} \leq \cdots \leq \lambda_{N}$. $\left[L_{G}\right]_{i j}$ stands for the element of row $i$ and column $j$ of matrix $L_{G}$. $\mathcal{N}_i$ denotes the set of neighbors of Agent $i$.
	
	Consider a MAS with the  connected undirected graph $\mathcal{G}$.  The dynamics of the $N$ agents are described by
	\begin{equation}\label{mainagent}
		\begin{aligned}
			&\left\{\begin{array}{l}
				x_{i}(k+1)=A x_{i}(k)+B u_{i}(k), \\
				y_{i}(k)=Cx_{i}(k),  k \in \mathbb{Z}_{\geq 0},
			\end{array}\right.
		\end{aligned}
	\end{equation}
	where $i=1,2,\cdots,N$, $A \in \mathbb{R}^{n \times n}, B \in \mathbb{R}^{n \times r}, C \in \mathbb{R}^{q \times n}$, $x_{i}(k) \in \mathbb{R}^{n}$ is the state information of Agent $i$, $u_{i}(k)  \in \mathbb{R}^{r}$ and $ y_{i} (k) \in \mathbb{R}^{q}$ are the control input and the output of Agent $i$, respectively.
	
	For the system (1), the state vector $x_i(k)$ is generally not directly measurable. Thus, we hope to use $y_i(k)$ to construct an observer state  to replace the real state to design a feedback control such that the consensus can be achieved. In this work,
	the observer is designed according to some rule,
	\begin{equation}\label{observer}
		\begin{aligned}
			\left\{\begin{array}{l}
				\hat{x}_{i}(k+1)=f(\hat{x}_{i}(k),u_i(k),y_i(k),\hat{y}_{i}(k)),\\
				\hat{y}_{i}(k)=g(\hat{x}_{i}(k)),
			\end{array}\right.
		\end{aligned}
	\end{equation}
	where $\hat{x}_{i}(k)$ and $ \hat{y}_{i}(k)$ denote the estimated value of ${x}_{i}(k)$ and ${y}_{i}(k)$, respectively. It is easy to see that the full-order state and the reduced-order state observers studied in Sections \ref{Sfull} and \ref{Sreduce}  can be considered as the special cases of the rule \eqref{observer}.

	In the traditional consensus algorithm, $u_i(k)$ is designed as follows
	\begin{equation*}
		\begin{aligned}
			u_{i}(k)=K\sum_{j \in \mathcal{N}_i} a_{i j}\left(\hat{x}_{j}(k)-\hat{x}_{i}(k)\right),
		\end{aligned}
	\end{equation*}
	where  $K \in \mathbb{R}^{r \times n}$ is the control gain matrix to be designed. However, if there is an semi-honest agent among these agents (or eavesdropper) who possesses certain information, including: (i) the dynamics of the agent, (ii) the model of the observer system, (iii) the consensus algorithm, (iv) the graph of the MAS, and (v) the estimated value $\hat{x}_{j}(k),j\in \mathcal{N}_i$, and the eavesdropper attempts to obtain the output information $y_i(k)$ of other agents, then the privacy  has a great risk of exposure.
	Thus, we add Laplace noise to the message sent by the Agent $i$ to others to preserve $y_i(k)$:
	\begin{equation}\label{lap}
		\begin{aligned}
			\theta_{i}(k)=\hat{x}_{i}(k)+\eta_{i}(k),
		\end{aligned}
	\end{equation} where $\eta_{i}(k) \sim \operatorname{Lap}(b_i(k))^n$ defined on a complete probability space $(\Omega, \mathcal{F},\{\mathcal{F}_{k}\}_{k \geq 0}, \mathbb{P})$,   $\mathcal{F}_{k}=\sigma\{\eta_{j}(s) \mid j=1,\cdots N$, $s<k\}$.
	
	Based on the information \eqref{lap}, the controller is given as follows
	\begin{equation}\label{control}
		\begin{aligned}
			u_{i}(k)=K\sum_{j \in \mathcal{N}_i} a_{i j}\left(\theta_{j}(k)-\hat{x}_{i}(k)\right).
		\end{aligned}
	\end{equation}
	
	{
		\begin{remark}
			The use of Laplace noise in this work follows the classical $\epsilon$-differential privacy framework. Compared with Gaussian noise, the main advantages of the Laplace mechanism are:
			(i) it provides {strict (pure)} $\epsilon$-differential privacy without requiring an additional $\delta$ parameter;
			(ii) its probability density function has a simple closed form, which enables a more tractable privacy analysis. iii) the work Wang et al. \cite{Wang2014} shows that Laplacian noise is optimal (among all possible distributions) in the sense that it minimizes the entropy of the transmitted messages while preserving differential privacy. 
			On the other hand,  Laplace noise has heavier tails than Gaussian noise, which may introduce larger variance into the exchanged information.
	\end{remark}}
	
	Denote $\eta(k)=[\eta_{1}^{\mathrm{T}}(k), \eta_{2}^{\mathrm{T}}(k),  \cdots, \eta_{N}^{\mathrm{T}}(k)]^{\mathrm{T}} \in \mathbb{R}^{Nn \times 1}$ and $\theta(k)=[\theta_{1}^{\mathrm{T}}(k), \theta_{2}^{\mathrm{T}}(k), \cdots, \theta_{N}^{\mathrm{T}}(k)]^{\mathrm{T}} \in \mathbb{R}^{Nn \times 1}$. Note that the noise is included in \eqref{control}. Then, stochastic consensus including the mean square and almost sure consensus should be investigated.
	
	\begin{definition}[Mean square consensus]
		The MAS \eqref{mainagent} under the observer \eqref{observer} and the control \eqref{control} is said to achieve the  mean square consensus if $\lim\limits _{k \rightarrow \infty}$ $ E\left\|x_{j}(k)-x_{i}(k)\right\|^2=0$ and  $\lim\limits _{k \rightarrow \infty} E\|x_i(k)-\hat{x}_i(k)\|^2=0$, for any given initial value and all $ i, j=1, \ldots, N$.
	\end{definition}
	
	\begin{definition}[Almost sure consensus]
		The MAS \eqref{mainagent} under the observer \eqref{observer} and the control \eqref{control} is said to achieve the  almost sure consensus  if $\lim\limits _{k \rightarrow \infty}$ $ \left\|x_{j}(k)-x_{i}(k)\right\|=0$ and  $\lim\limits _{k \rightarrow \infty} \|x_i(k)-\hat{x}_i(k)\|=0$, for any given initial value and all $ i, j=1, \ldots, N$.
	\end{definition}

	\
	
	In addition to obtaining the aforementioned consensus, we also ensure privacy preservation. We now give the definitions of adjacent and $\epsilon$-differential privacy for  general linear MASs.

	\begin{definition}[Adjacent]\label{adjacent}
		Denote $y(k)=[y_{1}^{\mathrm{T}}(k), y_{2}^{\mathrm{T}}(k), $ $y_{3}^{\mathrm{T}}(k), \cdots, y_{N}^{\mathrm{T}}(k)]^{\mathrm{T}} \in \mathbb{R}^{N \times q}$. Given $m>0$ and function $h(k)>0$, then two sets of output information $y(k)$ and $y^{\prime}(k)$ are adjacent with $m$ and $h(k)$ if, there exist $i_0 \in \{1,2,...,N \}$ and $k_{0} \geq 0$ such that $y_i(k)=y_i^{\prime}(k)$ for $i \neq i_0$ and $k \in \mathbb{Z}_{\geq 0}$, and
		\begin{equation}\label{adj}
			\begin{aligned}
				\begin{cases}y_{i_0}(k)=y^{\prime}_{i_0}(k), & k<k_{0} \\ \left\|y_{i_0}(k)-y_{i_0}^{\prime}(k)\right\|_1 \leq m h(k-k_{0}), & k \geq k_{0}.\end{cases}
			\end{aligned}
		\end{equation}
	\end{definition}
	
	\
	{
		\begin{remark}
			In our framework, the privacy of the \emph{entire output trajectories} is guaranteed through the definition of adjacency in Definition~\ref{adjacent} combined with the $\epsilon$-differential privacy mechanism. 	The adjacency relation is defined on the {whole output trajectories} rather 
			than on single time steps. Under such an adjacency notion, the $\epsilon$-differential privacy 
			guarantee ensures that an adversary cannot reliably distinguish between any two 
			such adjacent {entire trajectories}, thereby protecting the full process information 
			across all time steps. This modeling approach is consistent with  literature \cite{l1}.
	\end{remark}}
	
	\begin{remark}
		In previous works on observer-based privacy-preserving approaches, $\left\|y_{i_0}(k)-y_{i_0}^{\prime}(k)\right\|_1$ is required to be bounded by an exponentially decaying function\cite{l1}. However, in this work, we relax this constraint to enable privacy preservation across wider variations of adjacent datasets, which strengthens  data preservation.  {Moreover, $k_0$ allows the difference between adjacent datasets to begin at any time, rather than being restricted to starting at time 0. }
	\end{remark}

	\
	
	Note that \eqref{observer} - \eqref{control} tell us that $\theta(k)$ is determined by the noise $\eta(k)$ and the output $y(k)$  given the initial value   $\hat{x}(0)$. Then, we can define the function $\varphi: \mathbb{R}^{N \times n} \times \mathbb{R}^{N \times q} \rightarrow \mathbb{R}^{N \times n}$ such that $\varphi(\eta(k), y(k))=\theta(k)$ is well defined, which is consistent with related formulations in the literature \cite{22}.  {Let $O$ denote any Borel measurable set in the trajectory space $(\mathbb{R}^{n})^{\mathbb{N}}$.}  Hence, the $\epsilon$-differential privacy is defined as follows\cite{erfenwjm}.
	\begin{definition}[$\epsilon$-differential privacy]
		Let $\epsilon \geq 0$ be given. For $i \in\{1, \ldots, N \}$, $k \in \mathbb{Z}_{\geq 0}$, and any set $O \in \mathcal{B}((\mathbb{R}^{n})^\mathbb{N})$, the  $\epsilon$-differential privacy is said to be preserved if, for any adjacent  $y(k)$ and $y^{\prime}(k)$,
		$$
		\begin{aligned}
			\mathbb{P}\{{\eta}(k) \in \Omega \mid & \varphi ({\eta}(k), y(k)) \in O\} \\
			& \leq e^{\epsilon} \mathbb{P}\left\{{\eta}(k) \in \Omega \mid \varphi ({\eta}(k), y^{\prime}(k)) \in O\right\}.
		\end{aligned}
		$$
		
	\end{definition}
	
	{
		\begin{remark}
			Under the considered multi-agent system framework, the mapping $\varphi(\eta(k), y(k)) = \theta(k)$ is well-defined through a recursive causal mechanism. Starting from the initial condition $\hat{x}(0)$, the injection of noise $\eta(0)$ forms $\theta(0)$, which in turn determines the control input $u(0)$. Given $u(0)$, the known observer output $\hat{y}(0)$, and the measured output $y(0)$, the subsequent state $\hat{x}(1)$ is uniquely determined by the observer dynamics. This process repeats at each step $k$: from $\hat{x}(k)$ and $\eta(k)$ to $\theta(k)$, then to $u(k)$, and finally to $\hat{x}(k+1)$. This unbroken causal chain guarantees that the sequences $\eta(k)$ and $y(k)$ uniquely generate the entire trajectory of $\theta(k)$, thereby establishing a rigorous foundation for $\epsilon$-differential privacy analysis.
	\end{remark}}
	
	\
	
	The objective of this work is to construct the appropriate observer, controller $u_{i}(k)$, and Laplace noise $\eta(k)$ such that the MAS \eqref{mainagent} can achieve the above consensus and preserve the privacy simultaneously, which is said to achieve the differentially private consensus. In the following, we consider the observer-based differentially private consensus based on the full-order observer and the reduced-order observer, respectively. 	 
	
	\section{Differentially private consensus based on the full-order observer}\label{Sfull}
	In this section, we investigate differentially private consensus problem based on the following full-order observer:
	\begin{equation}\label{obser}
		\begin{aligned}
			\left\{\begin{array}{l}
				\hat{x}_{i}(k+1)=A \hat{x}_{i}(k)+B u_{i}(k)+L\left(y_{i}(k)-\hat{y}_{i}(k)\right), \\
				\hat{y}_{i}(k)=C\hat{x}_{i}(k) ,
			\end{array}\right.
		\end{aligned}
	\end{equation}
	where $ L \in \mathbb{R}^{n \times q}$ is the state estimation gain. 
	
	Denoting $\hat{x}(k)=\left[\hat{x}_{1}^{\mathrm{T}}(k), \hat{x}_{2}^{\mathrm{T}}(k), \cdots, \hat{x}_{N}^{\mathrm{T}}(k)\right]^{\mathrm{T}} \in \mathbb{R}^{N \times n}$ and the observer error $e_{i}(k)=x_{i}(k)-\hat{x}_{i}(k)$,  we have
	$$
	\begin{aligned}
		e_{i}(k+1)  =A e_{i}(k)-L  C  e_{i}(k)=(A-L C) e_{i}(k)
	\end{aligned}
	$$
	and
	$$
	\begin{aligned}
		x_{i}(k&+1) =A x_{i}(k)+B  K  \sum_{j \in \mathcal{N}_i} a_{i j}\left(\theta_{j}(k)-\hat{x}_{i}(k)\right) \\
		=&A x_{i}(k)-B K  \sum_{j=1}^{N}\left[L_{G}\right]_{i j}  \hat{x}_{j}(k)-B K \hspace{-0.2cm} \sum_{j=1, j \neq i}^{N}\hspace{-0.2cm}\left[L_{G}\right]_{i j} \eta_{j}(k).\\
	\end{aligned}
	$$
	We can rewrite the above equation as follows:
	\begin{equation}\label{comp1}
		\begin{aligned}
			x(k+1)&=\left[\left(I_{N} \otimes A\right)-\left(L_{G} \otimes B K\right)\right] x(k)\\
			&+\left(L_{G} \otimes B K\right) e(k) +(A_G \otimes B K) \eta(k),
		\end{aligned}
	\end{equation}
	where $x(k)=\left[x_{1}^{\mathrm{T}}(k), x_{2}^{\mathrm{T}}(k), \cdots, x_{N}^{\mathrm{T}}(k)\right]^{\mathrm{T}} \in \mathbb{R}^{N \times n}$ and $e(k)=\left[e_{1}^{\mathrm{T}}(k), e_{2}^{\mathrm{T}}(k),  \cdots, e_{N}^{\mathrm{T}}(k)\right]^{\mathrm{T}} \in \mathbb{R}^{N \times n}$.
	
	Defining the consensus error $\delta(k)=\left[\left(I_{N}-J_{N}\right) \otimes I_{n}\right] x(k)$ and the observer error $ \varrho (k)=[(I_{N}-J_{N}) \otimes I_{n}] e(k)$, we obtain
	$$
	\left\{\begin{aligned}
		\delta(k+1)= & {\left[\left(I_{N} \otimes A\right)-\left(L_{G} \otimes B K\right)\right] \delta(k) }\hspace{-0.1cm}+\hspace{-0.1cm}\left(L_{G} \otimes B K\right) \varrho (k) \\
		&+\left[\left(I_{N}-J_{N}\right) \otimes I_{n}\right]\left(A_{G} \otimes B K\right) \eta(k), \\
		\varrho (k+1)= & {\left[I_{N} \otimes(A-L C)\right] \varrho (k) , }\end{aligned}\right.
	$$
	where $J_N=\frac{1}{N}\mathbf{1}_{N} \mathbf{1}_{N}^{\mathrm{T}}$.
	
	Define the unitary matrix $\Psi=\left[\frac{\mathbf{1}_{N}}{\sqrt{N}}, \phi_{2}, \cdots, \phi_{N}\right]$ ,
	where $\phi_{i}$ is the unit eigenvector of ${L_G}$ associated with the eigenvalue $\lambda_{i}$.
	We can see that $\Psi^{\mathrm{T}} L_G \Psi=$ $\operatorname{diag}\left\{0, \lambda_{2}, \cdots, \lambda_{N}\right\}$.
	Defining $\delta(k)=\left(\Psi \otimes I_{n}\right) \widetilde{\delta}(k)$ and $\varrho (k)=\left(\Psi \otimes I_{n}\right) \widetilde{\varrho }(k)$, we have
	$$
	\begin{aligned}
		\widetilde{\varrho }(k+1)=  \left[I_{N} \otimes(A-L C)\right] \widetilde{\varrho }(k)
	\end{aligned}
	$$
	and
	$$
	\begin{aligned}
		\widetilde{\delta} (k+&1)=\left(\Psi \otimes I_{n}\right)^{\mathrm{T}}   \delta(k+1) \\
		=&\left[I_{N} \otimes A-\Delta \otimes B K\right] \widetilde{\delta}(k)+(\Delta \otimes B K)\widetilde{\varrho }(k)+M   \eta(k),
	\end{aligned}
	$$
	where $\Delta=\operatorname{diag}\left\{0, \lambda_{2},  \cdots, \lambda_{N}\right\}$ and $ M=\Psi^{\mathrm{T}}(A_{G}-$ $J_{N} A_{G}) \otimes B K $. Note that
	$$
	\begin{aligned}
		\widetilde{\delta}(k)=(\Psi^{\mathrm{T}} \otimes I_{n})   \delta(k) =[\Psi^{\mathrm{T}}(I_{N}-J_{N}) \otimes I_{n}] x(k).
	\end{aligned}
	$$
	It yields that
	$
	\widetilde{\delta}_{1}(k)=[\frac{I_{N}}{\sqrt{N}}\left(I_{N}-J_{N}\right) \otimes I_{n}] x(k)=0.
	$
	
	Denote $\xi(k)=\left[\widetilde{\delta}_{2}(k)^{\mathrm{T}}, \cdots, \widetilde{\delta}_{N}(k)^{\mathrm{T}}\right]^{\mathrm{T}}$, $\psi(k)=(\widetilde{\varrho }_{2}(k)^{\mathrm{T}},$ $ \cdots, \widetilde{\varrho }_{N}(k)^{\mathrm{T}})^{\mathrm{T}}$, and $\widetilde{\eta}(k)=[{\eta}_{2}(k)^{\mathrm{T}}, \cdots, {\eta}_{N}(k)^{\mathrm{T}}]^{\mathrm{T}}$. Then, we have
	\begin{equation}
		\left\{\begin{aligned}
			\xi(k+1)=&{R_{1}\xi(k)+R_{2}\psi(k) + \widetilde{M}  \widetilde{\eta}(k)},\\
			\psi(k+1)=&R_{3}{\psi}(k),\end{aligned}\right.
	\end{equation}
	where $R_{1}=I_{N-1} \otimes A-\Lambda \otimes B K$, $R_{2}=\Lambda \otimes B K$, $\Lambda=\operatorname{diag}\left\{\lambda_{2}, \cdots, \lambda_{N}\right\}$, $\widetilde{M}=\Phi^{\mathrm{T}}\left(A_{G}-J_{N} A_{G}\right) \otimes B K$, $\Phi=\left[\phi_{2}, \cdots, \phi_{N}\right]$, and $R_{3}=I_{N-1} \otimes(A-L C)$.
	
	\
	
	{ To quantify how fast the agents reach agreement, we give the following definition of the observer-based mean square convergence rate, extended from \cite{Mo2016}.}
	{
		\begin{definition}[Mean square convergence rate]
			Consider the multi-agent system \eqref{mainagent} under the observer \eqref{obser} and the control protocol \eqref{control}.
			The mean square convergence rate $\rho$ is defined as:
			\[
			\rho = \lim_{k \to \infty} \ \sup_{(\xi(0), \psi(0)) \neq 0} \left( \frac{\mathbb{E}[\|\xi(k)\|^2 ]}{\|\xi(0)\|^2 + \|\psi(0)\|^2} \right)^{\frac{1}{2k}}.
			\]
		\end{definition}
	}
	\
	{
		\begin{remark}
			The definition of the mean-square convergence rate in this paper differs from the case without an observer \cite{Mo2016}, fundamentally due to the coupling of system dynamics. When the system does not include an observer, the convergence rate only needs to consider the state error dynamics, and its definition can be based solely on the initial state $\xi(0)$. However, after introducing an observer, the evolution of the state error $\xi(k)$ not only depends on its own initial value but is also influenced by the observer error $\psi(k)$ through the coupling term $R_2\psi(k)$. Therefore, the worst-case convergence performance may be dominated by the initial condition of the observer error $\psi(0)$.
			The given definition can accurately reflect the worst-case convergence performance of the system under the influence of the observer.
		\end{remark}
	}
	\subsection{Consensus}\label{S31}
	In the following theorem, we give some sufficient conditions for the MAS \eqref{mainagent} based on the full-order observer \eqref{obser} to achieve the mean square and almost sure consensus. These  conditions imply that the separation principle still holds when examining  the consensus problem of linear MASs with decaying Laplace noise.  { Meanwhile, the convergence rate is provided. } 	The following  convergence theory for nonnegative almost super-martingales  will play an important role in obtaining the almost sure consensus.
	\begin{lemma}\label{semi} \cite{1985}
		Let $\{\mathcal{F}(k), k \in \mathbb{N}\}$ be a sequence of $\sigma$-algebras. Let $V(k)$, $\mu(k)$, $v(k)$, and $\omega(k)$ be $\mathcal{F}(k)$-measurable nonnegative random variables such that for all $k$,
		$$
		E\{V(k+1) \mid \mathcal{F}(k)\} \leq(1+\mu(k)) V(k)+v(k)-\omega(k) \text {, a.s. }
		$$
		If $\sum_{k=0}^{\infty} \mu(k)<\infty$ and $\sum_{k=0}^{\infty} v(k)<\infty$ a.s., then there exists a nonnegative random variable $x^{*}$ such that $\lim\limits _{k \rightarrow \infty} V(k)=x^{*}$ and $\sum_{k=0}^{\infty} \omega(k)<\infty$ a.s.
	\end{lemma}
	
	\
	
	\begin{theorem}\label{consen1}
		Consider $b_i(k)=c_{i} p_{i}(k)$ with $c_{i}>0$, $p_{i}(k)>0$, and $\sum_{k=0}^{\infty}p_{i}(k) < \infty$. If there exist two matrices $K$ and $L$ such that $\rho (A-L C)<1$ and $\rho (I_{N-1} \otimes A-\Lambda \otimes B K)< 1$ with $\Lambda=\operatorname{diag}\left\{\lambda_{2}, \cdots, \lambda_{N}\right\}$, then both the mean square consensus and  the almost sure consensus are achieved for  the MAS \eqref{mainagent} under the observer \eqref{obser} and the control \eqref{control}.
	\end{theorem}
	
	\begin{proof}
		First, we give the proof of the mean square consensus. Select $P_{1}$ such that $P_{1}^{-1} R_{1} P_{1}=J_{R}$, where $J_{R}$ is the Jordan normal form of $R_{1}$. Letting $\zeta(t)=P_{1}^{-1} \xi(t)$, one can see that
		$$
		\zeta(k+1)=J_{R} \zeta(k)+Z(k)+\widetilde{Z}(k),
		$$
		where $Z(k)=P_{1}^{-1} R_{2} \psi(k)$ and $\widetilde{Z}(k)=P_{1}^{-1} \widetilde{M} \widetilde{\eta}(k)$.
		Denote
		$$
		J_{\lambda_{q}, n q}=\left[\begin{array}{ccccc}
			\lambda_{q} & 1 & \cdots & 0 & 0\\
			0 & \lambda_{q} & \cdots & 0 & 0\\
			\vdots & \vdots & \ddots & \vdots & \vdots\\
			0 & 0 & \cdots & \lambda_{q} & 1\\
			0 & 0 & \cdots & 0 & \lambda_{q}\\
		\end{array}\right]
		$$
		as the $q$th Jordan block of $n_q$ size with eigenvalue $\lambda_{q}$. Denote $\zeta_{q}(k)=\left[\zeta_{q, 1}(k), \cdots,  \zeta_{q, j}(k), \cdots, \zeta_{q, n_{q}}(k)\right]^{\mathrm{T}}\in \mathbb{R}^{n_q}$, $Z_{q}(k)=\left[z_{q, 1}(k),\cdots,  z_{q, j}(k), \cdots, z_{q, n_{q}}(k)\right]^{\mathrm{T}}\in \mathbb{R}^{n_q}$, and $\widetilde{Z}_{q}(k)=\left[\widetilde{z}_{q, 1}(k),\cdots,  \widetilde{z}_{q, j}(k), \cdots\right.$, $\left.\widetilde{z}_{q, n_{q}}(k)\right]^{\mathrm{T}}\in \mathbb{R}^{n_q}$, where $\zeta_{q, j}, Z_{q, j}$, and $ \widetilde{Z}_{q, j}$ represent the $\left(\sum_{l=1}^{q-1} n_{l}+j\right)$th elements of $\zeta, Z$, and $\widetilde{Z}$, respectively. Then, we have
		$$
		\zeta_{q}(k+1)=J_{\lambda_{q}, n_{q}} \zeta_{q}(k)+Z_{q}(k)+\widetilde{Z}_{q}(k).
		$$
		It follows that
		\begin{equation}\label{zetaq}
			\left\{\begin{aligned}
				\zeta_{q, n_{q}}(k+1)=& \lambda_{q} \zeta_{q, n_{q}}(k)+Z_{q, n_{q}}(k)+\widetilde{Z}_{q, n_{q}}(k), \\
				\zeta_{q, j}(k+1)=& \lambda_{q} \zeta_{q, j}(k)+\zeta_{q, j+1}(k)\\
				& +Z_{q, j}(k)+\widetilde{Z}_{q, j}(k), j=1, \cdots, n_{q-1}.
			\end{aligned}\right.
		\end{equation}
		This can be considered as a semi-decoupled version.
		We can see from \eqref{zetaq} that
		$$
		\begin{aligned}
			\left\|\zeta_{q, n_{q}}(k+1)\right\|^{2} & =\lambda_{q}^{2}\left\|\zeta_{q, n_{q}}(k)\right\|^{2}+\left\|Z_{q, n_{q}}(k)\right\|^{2} \\
			& +2 \lambda_{q} \zeta_{q, n_{q}}^{\mathrm{T}}(k)\left(Z_{q, n_{q}}(k)+{\widetilde{Z}}_{q, n_{q}}(k)\right)\\
			& +2 Z_{q, n_{q}}^{\mathrm{T}}(k) \widetilde{Z}_{q, n_{q}}(k)+\left\|\widetilde{Z}_{q, n_{q}}(k)\right\|^{2}.
		\end{aligned}
		$$
		Taking the expectation on both sides, we have
		$$
		\begin{aligned}
			E\|\zeta_{q, n_q}&(k+1)\|^{2}= \lambda_{q}^{2} E\left\|\zeta_{q, n_{q}}(k)\right\|^{2}+E\left\|Z_{q, n_{q}}(k)\right\|^{2} \\
			& +E\left\|\widetilde{Z}_{q, n_{q}}(k)\right\|^{2}+2 \lambda_{q} E\left\{\zeta_{q, n_{q}}^{\mathrm{T}}(k) Z_{q, n_{q}}(k)\right\} \\
			\leq & \lambda_{q}{ }^{2} E\left\|\zeta_{q, n_{q}}(k)\right\|^{2}+E\left\|Z_{q, n_{q}}(k)\right\|^{2}+E\left\|\widetilde{Z}_{q, n_{q}(k)}\right\|^{2} \\
			& +\lambda_{q} \varepsilon E\left\|\zeta_{q, n_{q}}(k)\right\|^{2}+\frac{\lambda_{q}}{\varepsilon} E\left\|Z_{q, n_{q}}(k)\right\|^{2} \\
			= & \left(\lambda_{q}^{2}+\lambda_{q} \varepsilon\right) E\left\|\zeta_{q, n_{q}}(k)\right\|^{2}\\
			&+\left(1+\frac{\lambda_{q}}{\varepsilon}\right) E\left\|Z_{q, n_{q}}(k)\right\|^{2}+E\left\|\widetilde{Z}_{q, n_{q}}(k)\right\|^{2}.
		\end{aligned}
		$$
		Letting $r_{1}=\lambda_{q}^{2}+\lambda_{q} \varepsilon, r_{2}=1+\frac{\lambda q}{\varepsilon}$, we have
		$$
		\begin{aligned}
			E\|\zeta_{q, n_{q}}&(k+1)\|^{2} \leq r_{1}^{k+1} E\left\|\zeta_{q, n_{q}}(0)\right\|^{2} \\
			& +\sum_{l=0}^{k} r_{2} r_{1}^{k-l} E\left\|Z_{q, n_{q}}(l)\right\|^{2}+\sum_{l=0}^{k} r_{1}^{k-l} E\left\|\widetilde{Z}_{q, n q}(l)\right\|^{2}.
		\end{aligned}
		$$
		Since the spectral radius satisfies $\rho (I_{N-1} \otimes A - \Lambda \otimes BK) < 1$, there exists a sufficiently small $\varepsilon > 0$ ensuring $r_1 < 1$. It then follows that
		$$
		\begin{aligned}
			\lim\limits_{k \rightarrow \infty} r_1^{k+1} E\|\zeta_{q, n_q}(0)\|^2 = 0.
		\end{aligned}
		$$
		Since $\rho (A-L C)<1$, we have
		$$
		\begin{aligned}
			\lim\limits _{k \rightarrow \infty} E\|Z_{q, n_{q}}(k)\|^{2}=0
		\end{aligned}
		$$
		and
		$$
		\begin{aligned}
			\lim\limits _{k \rightarrow \infty} \sum_{l=0}^{k} r_{2} r_{1}{ }^{k-l} E\|Z_{q, n_{q}}(l)\|^{2}=0.
		\end{aligned}
		$$
		Since $\widetilde{Z}(k) = P_{1}^{-1} \widetilde{M} \widetilde{\eta}(k)$ is a linear combination of $\{\eta_i(k)\}$ and $\sum_{k=0}^{\infty} p_{i}(k) < \infty$, we obtain
		$$
		\begin{aligned}
			\lim\limits _{k \rightarrow \infty} E\|\widetilde{Z}(k)\|^{2}=0.
		\end{aligned}
		$$
		This immediately yields
		$$
		\begin{aligned}
			\lim\limits _{k \rightarrow \infty} E\|\widetilde{Z}_{q, n_{q}}(k)\|^{2}=0,
		\end{aligned}
		$$
		and  then
		$$
		\begin{aligned}
			\lim\limits _{k \rightarrow \infty} \sum_{l=0}^{k} r_{1}^{k-l} E\|\widetilde{Z}_{q, n_{q}}(l)\|^{2}=0.
		\end{aligned}
		$$
		Thus, we have
		\begin{equation}\label{sq23}
			\lim\limits _{k \rightarrow \infty} E\left\|\zeta_{q, n_{q}}(k+1)\right\|^{2}=0 .
		\end{equation}

		Now, we prove $\lim\limits _{k \rightarrow \infty} E\left\|\zeta_{q, j}(k+1)\right\|^{2}=0$ based on $\lim\limits _{k \rightarrow \infty} E\left\|\zeta_{q, j+1}(k+1)\right\|^{2}$ for $j=1,2, \cdots, n_{q}-1$, which is called the backstepping method. Observe that
		$$
		\begin{aligned}
			E\|\zeta_{q, j}&(k+1)\|^{2} \leq r_{1}^{k+1} E\left\|\zeta_{q, j}(0)\right\|^{2}+\sum_{l=0}^{k} r_{2} r_{1}^{k-l} E\left\|Z_{q, j}(l)\right\|^{2} \\
			& +\sum_{l=0}^{k} r_{1}^{k-l} E\left\|\widetilde{Z}_{q, j}(l)\right\|^{2}+\sum_{l=0}^{k} r_{1}^{k-l} E\left\|\zeta_{q, j+1}(l)\right\|^{2} .
		\end{aligned}
		$$
		Note that
		$$
		\lim\limits _{k \rightarrow \infty} \sum_{l=0}^{k} r_{1}{ }^{k-l} E\left\|\zeta_{q, j+1}(l)\right\|^{2}=0.
		$$
		Then, by the similar process in obtaining \eqref{sq23}, we have
		$$\lim\limits _{k \rightarrow \infty} E\|\zeta_{q, j}(k+1)\|^{2}=0, j=1,2, \cdots, n_{q}-1.$$
		That is, the mean square consensus is achieved.
		
		We now  provide the proof of the almost sure consensus.
		Letting $V(k)=\xi^{\mathrm{T}}(k) P \xi(k)$, we have
		$$
		\begin{aligned}
			V&(k+1)=\xi^{\mathrm{T}}(k+1) P \xi(k+1) \\
			& =\hspace{-0.1cm}\left[\hspace{-0.05cm}R_{1} \xi(k)\hspace{-0.08cm}+\hspace{-0.08cm}R_{2} \psi(k)\hspace{-0.12cm}+\hspace{-0.12cm}\widetilde{M} \widetilde{\eta}(k)\hspace{-0.05cm}\right]^{\mathrm{T}} \hspace{-0.1cm}P\hspace{-0.1cm}\left[\hspace{-0.05cm}R_{1} \xi(k)\hspace{-0.05cm}+\hspace{-0.12cm}R_{2} \psi(k)\hspace{-0.12cm}+\hspace{-0.1cm}\widetilde{M} \widetilde{\eta}(k)\hspace{-0.05cm}\right]  \\
			& =V(k)\hspace{-0.1cm}+\hspace{-0.1cm}\xi^{\mathrm{T}}(k)\left(R_{1}^{\mathrm{T}} P R_{1}\hspace{-0.1cm}-P\right) \xi(k)\hspace{-0.05cm}+\hspace{-0.05cm}\psi^{\mathrm{T}}(k) R_{2}^{\mathrm{T}} P R_{2} \psi(k) \\
			&\quad+\widetilde{\eta}^{\mathrm{T}}(k) \widetilde{M}^{\mathrm{T}} P \widetilde{M} \widetilde{\eta}(k) +2 \xi^{\mathrm{T}}(k) R_{1}^{\mathrm{T}} P R_{2} \psi(k)\\
			&\quad+2 \xi^{\mathrm{T}}(k) R_{1}^{\mathrm{T}} P \widetilde{M} \widetilde{\eta}(k)+2 \psi^{\mathrm{T}}(k) R_{2}^{\mathrm{T}} P \widetilde{M} \widetilde{\eta}(k) .
		\end{aligned}
		$$
		Taking the expectation on both sides, one has
		$$
		\begin{aligned}
			E\{V(k+&1) \mid \mathcal{F}_{k}\}=V_{k}+E\left\{\xi^{\mathrm{T}}(k)\left(R_{1}^{\mathrm{T}} P R_{1}-P\right) \xi(k)\right. \\
			& \left.+\psi^{\mathrm{T}}(k) R_{2}^{\mathrm{T}} P R_{2} \psi(k)+\widetilde{\eta}^{\mathrm{T}}(k) \widetilde{M}^{\mathrm{T}} P \widetilde{M} \widetilde{\eta}(k)\right\} \\
			&+2 \xi^{\mathrm{T}}(k) R_{1}^{\mathrm{T}} P R_{2} \psi(k)\\
			\leq & V_{k}+E\{\xi^{\mathrm{T}}(k)\left(R_{1}^{\mathrm{T}} P R_{1}-P\right) \xi(k) \\
			& +\psi^{\mathrm{T}}(k) R_{2}^{\mathrm{T}} P R_{2} \psi(k)+\lambda_{\max }(\widetilde{M}^{\mathrm{T}} P \widetilde{M})\|\widetilde{\eta}(k)\|^{2}\\
			&+\varepsilon_{1} \xi^{\mathrm{T}}(k) R_{1}^{\mathrm{T}} P R_{1} \xi(k)+\frac{1}{\varepsilon_{1}} \psi^{\mathrm{T}}(k) R_{2}^{\mathrm{T}} P R_{2} \psi(k)\} \\
			=&V_{k}+\lambda_{\max }(\widetilde{M}^{\mathrm{T}} P \widetilde{M})n \sum_{i=1}^{N} c_{i} p_{i}(k)\\
			&+E\left\{\xi^{\mathrm{T}}(k)\left(R_{1}^{\mathrm{T}} P R_{1}-P+\varepsilon_{1} R_{1}^{\mathrm{T}} P R_{1}\right) \xi(k)\right.\\
			&+\left.\psi^{\mathrm{T}}(k)\left(R_{2}^{\mathrm{T}} P R_{2}+\frac{1}{\varepsilon_{1}} R_{2}^{\mathrm{T}} P R_{2}\right) \psi(k)\right\}.
		\end{aligned}
		$$
		Since $\rho\left(R_{3}\right)<1$, we have
		$$
		\begin{aligned}
			\lim\limits _{k \rightarrow \infty} \|\psi(k)\|=\lim\limits _{k \rightarrow \infty}\|x_i(k)-\hat{x}_i(k)\|=0,
		\end{aligned}
		$$
		$$
		\begin{aligned}
			\sum_{k=0}^{\infty}\|\psi(k)\|^{2}<\infty,
		\end{aligned}
		$$
		and
		$$
		\begin{aligned}
			\sum_{k=0}^{\infty}& \psi^{\mathrm{T}}(k)\left(1+\frac{1}{\varepsilon_{1}}\right) R_{2}^{\mathrm{T}} P R_{2} \psi(k) \\
			&< \left(1+\frac{1}{\varepsilon_{1}}\right) \lambda_{\max }(R_{2}^{\mathrm{T}} P R_{2}) \sum_{k=0}^{\infty}\|\psi(k)\|^{2}<\infty.
		\end{aligned}
		$$
		Note that $$\sum_{k=0}^{\infty} \lambda_{\max}(\widetilde{M}^{\mathrm{T}} P \widetilde{M}) n \sum_{i=1}^{N} c_{i} p_{i}(k) < \infty.$$ Since $\rho\left(R_{1}\right)<1$, there exists a positive definite matrix $P$ such that $R_{1}^{\mathrm{T}} P R_{1}-P+Q=0$, where $Q>0$. Select a small enough $\varepsilon_{1}$ such that $\varepsilon_{1} R_{1}^{\mathrm{T}} P R_{1}<Q$. By Lemma \ref{semi}, we have $\lim\limits _{k \rightarrow \infty} V(k)$ exists and is finite. Since $\lim\limits_{k \rightarrow \infty} E\{V(k)\}=0$, we must have $$\lim\limits _{k \rightarrow \infty} V(k)=0 \quad a.s.,$$ which implies $$\lim\limits _{k \rightarrow \infty} \|\xi_{k}\|=0 \quad a.s.$$ Now, the proof is completed.
	\end{proof}
	
	\
	
	\begin{remark}
		Evidently, Theorem 1 relaxes the constraints on the variance of Laplace noise, enabling the achievement of both mean square consensus and almost sure consensus without requiring exponentially decaying noise variance(\cite{22,81}).
	\end{remark}
	
	{	\begin{remark}[Separation principle]
			The separation principle refers to the classical result stating that, for linear systems, the controller design and the observer design can be carried out independently. 
			Specifically, the stability of the closed-loop system can be ensured by combining a stabilizing state-feedback controller with a stable observer, without requiring a joint or coupled design procedure. 
			In  this paper, the separation principle means that the consensus protocol and the observer (with decaying Laplace noise) can be designed separately, and the overall closed-loop consensus performance remains guaranteed.
	\end{remark}}
	
	\
	
	Based on the aforementioned theorem, we can get the following conclusion.
	\begin{corollary}\label{con1}
		Consider $b_i(k)=c_{i} g_{i}^{k}$ with $c_{i}>0$ and $0<g_{i}<1$. If there exist two matrices $K$ and $L$ such that $\rho (A-L C)<1$ and $\rho (I_{N-1} \otimes A-\Lambda \otimes B K)< 1$ with $\Lambda=\operatorname{diag}\left\{\lambda_{2}, \cdots, \lambda_{N}\right\}$, then both the mean square consensus and  the almost sure consensus are achieved for  the MAS \eqref{mainagent} under the observer \eqref{obser} and the control \eqref{control}.
	\end{corollary}
	

	{
		From Corollary \ref{con1}, we have established that the system achieves consensus when the scale parameter $b_i(k)$ decays exponentially. A natural and important question that follows is: what is the explicit convergence rate in the mean square sense under such conditions? The following theorem addresses this question by  characterizing the convergence rate. 
	}
	{
		\begin{theorem}\label{mscr1}
			Consider $b_i(k) = c_i g_i^k$ with $c_i > 0$ and $0 < g_i < 1$. Assume there exist matrices $K$ and $L$ such that $\rho(A - LC) < 1$ and $\rho(I_{N-1} \otimes A - \Lambda \otimes BK) < 1$, where $\Lambda = \operatorname{diag}\{\lambda_2, \dots, \lambda_N\}$. Then, the mean-square convergence rate is given by:
			\[
			\rho = \max\left( \rho(I_{N-1} \otimes A - \Lambda \otimes BK), \rho(A - LC), \max_i g_i \right).
			\]
		\end{theorem}
	}
	
	\begin{proof}
		Using the transformed coordinates, we have the system:
		\[
		\begin{cases}
			\xi(k+1) = R_1 \xi(k) + R_2 \psi(k) + \widetilde{M} \widetilde{\eta}(k) \\
			\psi(k+1) = R_3 \psi(k)
		\end{cases}
		\]
		where $R_1 = I_{N-1} \otimes A - \Lambda \otimes BK$, $R_2 = \Lambda \otimes BK$, $R_3 = I_{N-1} \otimes (A - LC)$.
		
		The solution for $\psi(k)$ is $\psi(k) = R_3^k \psi(0)$. Substituting into the $\xi$-dynamics and solving recursively yields:
		\[
		\xi(k) = R_1^k \xi(0) + \sum_{j=0}^{k-1} R_1^{k-1-j} R_2 R_3^j \psi(0) + \sum_{j=0}^{k-1} R_1^{k-1-j} \widetilde{M} \widetilde{\eta}(j).
		\]
		
		Define the deterministic and stochastic parts as:
		\begin{align*}
			d_{\xi}(k) &= R_1^k \xi(0) + \sum_{j=0}^{k-1} R_1^{k-1-j} R_2 R_3^j \psi(0), \\
			n_{\xi}(k) &= \sum_{j=0}^{k-1} R_1^{k-1-j} \widetilde{M} \widetilde{\eta}(j).
		\end{align*}
		Then, $\xi(k) = d_{\xi}(k) + n_{\xi}(k)$. Since $\widetilde{\eta}(k)$ is zero-mean and independent of initial conditions, we have
		\[
		\mathbb{E}[\xi(k)^T \xi(k)] = \mathbb{E}[d_{\xi}(k)^T d_{\xi}(k)] + \mathbb{E}[n_{\xi}(k)^T n_{\xi}(k)].
		\]
		
		Let $\rho_1 = \rho(R_1) = \rho(I_{N-1} \otimes A - \Lambda \otimes BK)$, $\rho_2 = \rho(R_3) = \rho(A - LC)$, and $\rho_s = \max(\rho_1, \rho_2)$. There exists a constant $\mathcal{C} > 0$ such that for all $t \ge 0$, $\|R_1^t\| \le \mathcal{C} \rho_1^t$ and $\|R_3^t\| \le \mathcal{C} \rho_2^t$.
		
		For the deterministic part,
		\begin{align*}
			\|d_{\xi}(k)\| &\le \|R_1^k \xi(0)\| + \sum_{j=0}^{k-1} \|R_1^{k-1-j} R_2 R_3^j \psi(0)\| \\
			&\le \mathcal{C} \rho_1^k \|\xi(0)\| + \sum_{j=0}^{k-1} \mathcal{C} \rho_1^{k-1-j} \|R_2\| \mathcal{C} \rho_2^j \|\psi(0)\|.
		\end{align*}
		Thus,
		\[
		\|d_{\xi}(k)\| \le \mathcal{C} \rho_1^k \|\xi(0)\| + \mathcal{C}^2 \|R_2\| \|\psi(0)\| \sum_{j=0}^{k-1} \rho_1^{k-1-j} \rho_2^j.
		\]
		Note that $\sum_{j=0}^{k-1} \rho_1^{k-1-j} \rho_2^j \le k \rho_s^{k-1}$. Applying the inequality $(a + b)^2 \le 2a^2 + 2b^2$, we have
		\[
		\|d_{\xi}(k)\|^2 \le 2\mathcal{C}^2 \rho_1^{2k} \|\xi(0)\|^2 + 2\mathcal{C}^4 \|R_2\|^2 \|\psi(0)\|^2 k^2 \rho_s^{2k-2}.
		\]
		Since $\psi(0)$ is fixed, $\|\psi(0)\|$ is finite. Therefore, we obtain
		\[
		\frac{\|d_{\xi}(k)\|^2}{\|\xi(0)\|^2 + \|\psi(0)\|^2} \le 2\mathcal{C}^2 \rho_1^{2k} + \frac{2\mathcal{C}^4 \|R_2\|^2 \|\psi(0)\|^2 k^2 \rho_s^{2k-2}}{\|\xi(0)\|^2 + \|\psi(0)\|^2}.
		\]
		Taking the $k$-th root and limit as $k \to \infty$, we have
		\[
		\limsup_{k \to \infty} \left( \frac{\|d_{\xi}(k)\|^2}{\|\xi(0)\|^2 + \|\psi(0)\|^2} \right)^{1/k} \le \rho_s^2.
		\]
		
		For the stochastic part,
		\[
		\mathbb{E}[n_{\xi}(k)^T n_{\xi}(k)] = \sum_{j=0}^{k-1} \operatorname{tr} \left[ R_1^{k-1-j} \widetilde{M} \Lambda_j \widetilde{M}^T (R_1^{k-1-j})^T \right],
		\]
		where $\Lambda_j = \operatorname{diag}(2c_1^2 g_1^{2j}, 2c_1^2 g_2^{2j}, \dots)$. Let $\theta = \max_i g_i^2$ and $c = \max_i c_i$. Then $\Lambda_j \leq 2c^2 \theta^j I$, and
		\begin{align*}
			\mathbb{E}[n_{\xi}(k)^T n_{\xi}(k)] &\le 2c^2 \sum_{j=0}^{k-1} \theta^j \| R_1^{k-1-j} \widetilde{M} \|_F^2 \\
			&\le 2c^2 \mathcal{C}^2 \|\widetilde{M}\|_F^2 \sum_{j=0}^{k-1} \theta^j \rho_1^{2(k-1-j)}.
		\end{align*}
		Hence,
		\[
		\mathbb{E}[n_{\xi}(k)^T n_{\xi}(k)] \le 2c^2 \mathcal{C}^2 \|\widetilde{M}\|_F^2 \rho_1^{2(k-1)} \sum_{j=0}^{k-1} \left( \frac{\theta}{\rho_1^2} \right)^j.
		\]
		The growth rate satisfies
		\[
		\limsup_{k \to \infty} \left( \mathbb{E}[n_{\xi}(k)^T n_{\xi}(k)] \right)^{1/k} \le \max(\theta, \rho_1^2).
		\]
		Since $\mathbb{E}[n_{\xi}(k)^T n_{\xi}(k)]$ is independent of initial conditions, we have
		\[
		\frac{\mathbb{E}[n_{\xi}(k)^T n_{\xi}(k)]}{\|\xi(0)\|^2 + \|\psi(0)\|^2} = \mathcal{O}\left( [\max(\theta, \rho_1^2)]^k \right).
		\]
		
		Combining both parts, we obtain
		\begin{align*}
			\frac{\mathbb{E}[\xi(k)^T \xi(k)]}{\|\xi(0)\|^2 + \|\psi(0)\|^2}
			&\le \mathcal{O}(\rho_s^{2k}) + \mathcal{O}\left( [\max(\theta, \rho_1^2)]^k \right) \\
			&= \mathcal{O}\left( [\max(\rho_s^2, \theta, \rho_1^2)]^k \right).
		\end{align*}
		Therefore,
		\begin{align*}
			\rho &\le \max(\rho_s, \sqrt{\theta}, \rho_1) = \max(\rho_s, \sqrt{\theta}) \\
			&= \max\left( \rho(I_{N-1} \otimes A - \Lambda \otimes BK), \rho(A - LC), \max_i g_i \right).
		\end{align*}
		The equality is achieved by considering initial conditions aligned with the slowest-decaying modes. This completes the proof.
	\end{proof}
	
	\
	
	{
		\begin{remark}
			Theorem \ref{mscr1} reveals that the overall mean-square convergence rate is determined by the worst-case (slowest) convergence behavior among three fundamental factors: (i) the inter-agent consensus dynamics $\rho(I_{N-1} \otimes A - \Lambda \otimes BK)$, (ii) the estimation error dynamics $\rho(A - LC)$, and (iii) the maximum decay rate of the privacy-preserving noise sequences $\max_i g_i$.  In particular, the result highlights a fundamental trade-off in privacy-preserving consensus-while stronger privacy protection typically requires slower noise decay (larger $g_i$), this may comes at the direct cost of reduced convergence rate.
	\end{remark}}
	\begin{remark}
		In \cite{97}, Liu et al. investigated the differentially private output consensus problem for continuous-time heterogeneous linear MASs and got static average consensus in mean square sense. Different from  \cite{97}, we study the observer-based differentially private consensus based on the full-order
		observer and obtain both the mean square consensus and the almost sure consensus of the observer-based discrete-time general linear MASs.
	\end{remark}
	\begin{remark}
		Note that the existence of $L$ such that   $\rho (A-L C)<1$ can be ensured by the condition that $(A,C)$ is detectable.  If $\gamma^{*}={argmin}_{\gamma}[\exists \hat{\mathcal{P}} >0 \mid \hat{\mathcal{P}}> A^{\mathrm{T}} \hat{\mathcal{P}} A+I_{n}-\gamma A^{\mathrm{T}} \hat{\mathcal{P}} B(I_{m}+B^{\mathrm{T}} \hat{\mathcal{P}} B)^{-1} B^{\mathrm{T}} \hat{\mathcal{P}} A]\in(0,1)$ is well defined, then for any $1>\gamma \geq \gamma^{*}$, the sufficient condition $\rho (I_{N-1} \otimes A-\Lambda \otimes B K)< 1$ can be guaranteed with
		\begin{equation*}
			K=\varsigma \left(I_{r}+B^{\mathrm{T}} \mathcal{P} B\right)^{-1}B^{\mathrm{T}} \mathcal{P}A,
		\end{equation*}
		where
		\begin{equation*}
			-\frac{\sqrt{1-\gamma}}{\lambda_{2}}+\frac{1}{\lambda_{2}} \leq \varsigma  \leq \frac{\sqrt{1-\gamma}}{\lambda_{N}}+\frac{1}{\lambda_{N}}
		\end{equation*}
		and $\mathcal{P}$ satisfies the following parameterized generalized Algebraic Riccati Equation:
		\begin{eqnarray*}
			0&=& A^{\mathrm{T}} \mathcal{P}A-\gamma A^{\mathrm{T}}\mathcal{P} B \cr
			&& \times\left(I_{m}+B^{\mathrm{T}} \mathcal{P} B\right)^{-1} B^{\mathrm{T}} \mathcal{P}A+I_{n}-\mathcal{P}.
		\end{eqnarray*}
		See \cite{occ} for more details about the solution to the above  parameterized generalized Algebraic Riccati Equation.
	\end{remark}

	From the established sufficient conditions of Theorem \ref{consen1},  the  state estimation gain $L$ and the feedback gain  $K$ can be independently designed, which ensures both the mean square stability and the almost sure stability of the agents' estimation errors and pairwise measurement errors, even in the presence of Laplace noise. This foundational concept underpins the separation principle for observer-based consensus in linear MASs. Furthermore, this principle extends to the observer-based stabilization of a single linear system. To elucidate this concept, consider the following linear system accompanied by its full-order observer:
	\begin{equation}\label{singles}
		\begin{aligned}
			&\left\{\begin{array}{l}
				x(k+1)=A x(k)+B u(k), \\
				y(k)=Cx(k),  k \in \mathbb{Z}_{\geq 0}
			\end{array}\right.
		\end{aligned}
	\end{equation}
	and
	\begin{equation}\label{singlesobser}
		\begin{aligned}
			\left\{\begin{array}{l}
				\hat{x}(k+1)=A \hat{x}(k)+B u(k)+L\left(y(k)-\hat{y}(k)\right), \\
				\hat{y}(k)=C\hat{x}(k) ,
			\end{array}\right.
		\end{aligned}
	\end{equation}
	where
	\begin{equation}\label{feedbacks}
		u(k)=K(\hat{x}(k)+\eta(k))
	\end{equation}
	and $\eta(k) \sim \operatorname{Lap}(b(k))^n$. Here, $u(k)$ is control input based on the observer \eqref{singlesobser} affected by the noise $\eta(k)$. Theorem \ref{consen1} implies the following corollary, which is also new.
	\begin{corollary}\label{corostab}
		Consider $b(k)=c p^{k}$ with $c>0$ and $0<p<1$. If there exist two matrices $K$ and $L$ such that $\rho (A-L C)<1$ and $\rho ( A- B K)< 1$, then both the mean square  and  the almost sure stabilization are achieved for    \eqref{singles} under  the control \eqref{feedbacks}.
	\end{corollary}
	
	\

	\subsection{Differential privacy}\label{S33}
	Next, we give the differential privacy theorem for the MAS \eqref{mainagent} under the full-order observer \eqref{obser}.
	\begin{theorem}\label{DP1}
		Consider $b_i(k)=c_{i} p_{i}(k)$ with $c_{i}>0$ and $p_i(k)>0$. Let  the observer's initial values $\hat{x}_{y, i}(0)=\hat{x}_{y^{\prime}, i}(0)$ for adjacent  $y(k)$ and $y^{\prime}(k)$ with $m>0$ and  $h(k)>0$.
		If there exist functions $h(k)$ and $p_{i}(k)$, and matrices $K$ and $L$ such that
		\begin{equation}\label{hp}
			\epsilon =\max_{1\le i \le N}\left(\|L\|_1  m \sum_{k=1}^{\infty }\frac{\sum_{b=0}^{k-1}l_{i}^{k-b-1}h(b)}{c_i p_i(k) } \right)< \infty,
		\end{equation}
		then the  $\epsilon$-differential privacy is preserved, where $l_i=\left\|A-L C-d_i B K\right\|_1$.
	\end{theorem}

	\begin{proof}
		It can be seen that
		$$
		\begin{aligned}
			P\{\varphi(\eta(k), y(k)) \in O\}=&P\{\hat{x}_{y, i}(k)+\eta_{i}(k) \in O\}\\
			=&P\{\eta_{i}(k) \in O-\hat{x}_{y,i}(k)\}.
		\end{aligned}
		$$
		Letting $ \beta  _{i}(k)=\hat{x}_{y,i}(k)-\hat{x}_{y^\prime, i}(k)$, we have
		$$
		\begin{aligned}
			P\{\varphi(\eta(k), y^{\prime}(k)) \in O\}&=P\{\hat{x}_{y^\prime, i}(k)+\eta_{i}(k) \in O\} \\
			&=P\{\hat{x}_{y, i}(k)- \beta  _{i}(k)+\eta_{i}(k) \in O\}\\
			&=P(\eta_{i}(k)- \beta  _{i}(k) \in O-\hat{x}_{y,i}(k)\}.
		\end{aligned}
		$$
		Thus, we obtain
		$$
		\begin{aligned}
			&\frac{\mathbb{P}\left\{\boldsymbol{\eta}(k) \in \Omega \mid \varphi (\boldsymbol{\eta}(k), y(k)) \in O\right\}}{ \mathbb{P}\left\{\boldsymbol{\eta}(k) \in \Omega \mid \varphi (\boldsymbol{\eta}(k), y^{\prime}(k)) \in O\right\}} \\
			&\quad =\prod_{i\in \mathcal{V}}\prod_{k=0}^{\infty}\frac{\int_{O-\hat{x}_{y,i}(k)}L(\eta_i(k),0,b_i(k))d\eta_i(k)}{\int_{O-\hat{x}_{y,i}(k)}L(\eta_i(k),- \beta  _i(k),b_i(k))d\eta_i(k)}   \\
			&\quad =\prod_{i\in \mathcal{V}}\prod_{k=0}^{\infty}e^{ -\frac{\|\eta_i(k) \|_1}{b_i(k)}+\frac{\| \eta_i(k)+ \beta  _i(k)\|_1}{b_i(k)}  }\\
			& \quad\leq \prod_{i\in \mathcal{V}}\prod_{k=0}^{\infty}e^{ \frac{\| \beta  _i(k) \|_1}{b_i(k)} }\\
			& \quad= e^{\sum_{i\in \mathcal{V}}\sum_{k=0}^{\infty}\frac{\| \beta  _i(k) \|_1}{b_i(k)}}.
		\end{aligned}
		$$
		
		It can be seen that
		$$
		\left\{\begin{aligned}
			\hat{x}_{y, i}(k+1)=&\left(A-L C-d_i B K\right)   \hat{x}_{y, i}(k)\\
			&+L   y_{i}(k)+B K   \sum_{j \in \mathcal{N}_i} a_{i j}   \theta_{j}(k), \\
			\hat{x}_{y^{\prime}, i}(k+1)=&(A-L C-d_i B K)   \hat{x}_{y^{\prime}, i}(k)\\
			&+L   y_{i}^{\prime}(k)+B K \sum_{j \in \mathcal{N}_i} a_{i j}   \theta_{j}^{\prime}(k).
		\end{aligned}\right.
		$$
		Noting $O=\left\{\theta_{i}(k), i\in \mathcal{V}\right\}$ and $O'=\left\{\theta'_{i}(k), i\in \mathcal{V}\right\}$ for $y$ and $y^{\prime}$ are the same, we have
		$$
		\begin{aligned}
			\hat{x}_{y, i}&(k+1)-\hat{x}_{y^{\prime}, i}(k+1)\\
			&=\left(A\hspace{-0.1cm}-\hspace{-0.1cm}L C\hspace{-0.1cm}-\hspace{-0.1cm}d_i B K\right)\left[\hat{x}_{y, i}(k)-\hat{x}_{y^{\prime}, i}(k)\right]\hspace{-0.1cm}+\hspace{-0.1cm}L\left[y_{i}(k)-y_{i}^{\prime}(k)\right].
		\end{aligned}
		$$
		Since $\hat{x}_{y, i}(0)=\hat{x}_{y^{\prime}, i}(0)$, it follows that
		$$
		\begin{aligned}
			\hat{x}_{y, i}(k+1)-&\hat{x}_{y^{\prime}, i}(k+1)\\
			&=\sum_{b=0}^{k}\left(A-L C-d_i B K\right)^{k-b}   L  \left[y_{i}(b)-y_{i}^{\prime}(b)\right]
		\end{aligned}
		$$
		and
		$$
		\begin{aligned}
			\hat{x}_{y,i}&(k)-\hat{x}_{y^{\prime}, i}(k)\\
			&=\sum_{b=0}^{k-1}\left(A-L C-d_i B K\right)^{k-b-1}   L  \left[y_{i}(b)-y_{i}^{\prime}(b)\right],k\ge 1.
		\end{aligned}
		$$
		For $i=i_0$, we have
		$$
		\begin{aligned}
			\|\hat{x}&_{y,i_0}(k)-\hat{x}_{y^{\prime},i_0}(k)\|_1 \\
			& \leq\|L\|_1 \sum_{b=0}^{k-1}\left\|A-L C-d_{i_0} B K\right\|_1^{k-b-1} \|y_{i_0}(b)-y_{i_0}^{\prime}(b)\|_1 \\
			& \leq\|L\|_1   \sum_{b=k_{0}}^{k-1}l_{i_0}^{k-b-1}   m h(b-k_0) \\
			& \leq\|L\|_1   m   \sum_{b=0}^{k-1}l_{i_0}^{k-b-1}    h(b) .\\
		\end{aligned}
		$$
		For $i\ne i_0$, we have $$\|\hat{x}_{y,i}(k)-\hat{x}_{y^{\prime},i}(k)\|_1=0.$$
		If there exist functions $h(k)$ and $p_{i}(k)$ and matrices $K$ and $L$ such that \eqref{hp} holds,  then we have
		\begin{equation}\label{pp}
			\begin{aligned}
				&\frac{\mathbb{P}\{\boldsymbol{\eta}(k) \in \Omega \mid \varphi (\boldsymbol{\eta}(k), y(k)) \in O\}}{ \mathbb{P}\left\{\boldsymbol{\eta}(k) \in \Omega \mid \varphi  (\boldsymbol{\eta}(k), y^{\prime}(k)) \in O\right\}}\\
				&\qquad \leq e^{\|L\|_1   m \sum_{k=1}^{\infty }\frac{\sum_{b=0}^{k-1}l_{i_0}^{k-b-1}h(b)}{c_{i_0} p_{i_0}(k)}} \\
				&\qquad \leq e^{\max_{1\le i \le N}\left(\|L\|_1  m \sum_{k=1}^{\infty }\frac{\sum_{b=0}^{k-1}l_{i}^{k-b-1}h(b)}{c_i p_i(k) } \right)} < \infty.
			\end{aligned}
		\end{equation}
		Letting $\epsilon =\max_{1\le i \le N}\left(\|L\|_1  m \sum_{k=1}^{\infty }\frac{\sum_{b=0}^{k-1}l_{i}^{k-b-1}h(b)}{c_i p_i(k) } \right)$, the  $\epsilon$-differential privacy is preserved. The proof is completed.
	\end{proof}
	
	\
	
	Theorem \ref{DP1} presents a complex integrated condition involving adjacency distance, noise design, and gain design. In fact, this theorem relaxes the restrictions on the noise variance and the definition of adjacent. Below, we derive several interesting corollaries. 
	\begin{corollary}\label{con2}
		Consider $b_i(k)=c_{i}\frac{1}{(k+1)^2}$, i.e., $p_i(k)=\frac{1}{(k+1)^2}$. Let the observer's initial values $\hat{x}_{y, i}(0)=\hat{x}_{y^{\prime}, i}(0)$ for adjacent  $y(k)$ and $y^{\prime}(k)$ with $m>0$,  $h(k)>0$, and
		\begin{equation}\label{hbb}
			\sum_{k=0}^{\infty}h(k)k^2 < \infty.
		\end{equation}
		If there exist matrices $K$ and $L$ such that $l_i=\left\|A-L C-d_i B K\right\|_1<1$, then the  $\epsilon$-differential privacy is preserved with $\epsilon=\max_{1\le i \le N}\epsilon_i$, where $\epsilon_i=\|L\|_1  m \sum_{b=0}^{\infty} h(b) \frac{\left[ (b + 2)^2 - (2b^2 + 6b + 3) l_i + (b + 1)^2 l_i^2 \right]}{c_{i}(1 - l_i)^3}$.
	\end{corollary}
	
	\begin{proof}
		Letting $p_i(k)=\frac{1}{(k+1)^2}$, we have
		\begin{equation*}
			\begin{aligned}
				\epsilon_i= e^{\|L\|_1 m {c_{i}}^{-1} \sum_{k=1}^{\infty }(k+1)^2\sum_{b=0}^{k-1}l_{i}^{k-b-1}h(b)}.
			\end{aligned}
		\end{equation*}
		
		It can be seen that
		$$
		\begin{aligned}
			\sum_{k=1}^{\infty }&(k+1)^2\sum_{b=0}^{k-1}l_{i}^{k-b-1}h(b)\\
			&=\sum_{b=0}^{\infty} h(b) \frac{\left[ (b + 2)^2 - (2b^2 + 6b + 3) l_{i} + (b + 1)^2 l_{i}^2 \right]}{(1 - l_{i})^3},
		\end{aligned}
		$$
		which is  finite under the condition \eqref{hbb}. Letting $$\epsilon_i=\|L\|_1  m \hspace{-0.1cm}  \sum_{b=0}^{\infty} \hspace{-0.1cm}h(b) \frac{\left[ (b + 2)^2 \hspace{-0.1cm}-\hspace{-0.1cm} (2b^2 + 6b + 3) l_i \hspace{-0.1cm}+\hspace{-0.1cm} (b + 1)^2 l_i^2 \right]}{c_{i}(1 - l_i)^3},$$
		the $\epsilon$-differential privacy is preserved  with  $\epsilon=\max_{1\le i \le N}\epsilon_i$. The proof is completed.
	\end{proof}
	
	\
	
	\begin{corollary}\label{con3}
		Consider $b_i(k)=c_{i} g_{i}^{k}$, i.e., $p_i(k)=g_{i}^{k}$, $0<g_i<1$. Let the observer's initial values $\hat{x}_{y, i}(0)=\hat{x}_{y^{\prime}, i}(0)$ for adjacent  $y(k)$ and $y^{\prime}(k)$ with $m>0$ and  $h(k)>0$. Assume that there exist constant $\mathcal{C}$ and real number $\alpha < g_i$ such that
		\begin{equation}\label{ca}
			h(k) \leq \mathcal{C} \alpha^k
		\end{equation}
		for all sufficiently large $k$. If there exist matrices $K$ and $L$ such that $l_i=\left\|A-L C-d_i B K\right\|_1< g_i$, then the  $\epsilon$-differential privacy is preserved with $\epsilon=\max_{1\le i \le N}\epsilon_i$, where $\epsilon_i=\frac{\|L\|_1  m  }{c_{i}(g_{i} - l_i)} \sum_{b=0}^{\infty} h(b) \left( \frac{1}{g_{i}} \right)^b $.
	\end{corollary}
	
	\begin{proof}
		Letting $p_i(k)=g_{i}^{k}$, we have
		\begin{equation}\label{}
			\begin{aligned}
				\epsilon =\max_{1\le i \le N}\left (\|L\|_1 m {c_{i}}^{-1} \sum_{k=1}^{\infty }g_{i}^{-k}\sum_{b=0}^{k-1}l_{i}^{k-b-1}h(b)\right).
			\end{aligned}
		\end{equation}
		By exchanging the order of summation, we have
		$$
		\begin{aligned}
			\sum_{k=1}^{\infty }g_{i}^{-k}\sum_{b=0}^{k-1}l_{i}^{k-b-1}h(b)&= \frac{1}{l_{i}} \sum_{b=0}^{\infty} \frac{h(b)}{l_{i}^b} \sum_{k=b+1}^{\infty} \left( \frac{l_{i}}{g_{i}} \right)^k \\
			&=\frac{1}{g_{i} - l_{i}} \sum_{b=0}^{\infty} h(b) \left( \frac{1}{g_{i}} \right)^b.
		\end{aligned}
		$$
		Therefore, the convergence of the original series is equivalent to that of the series \(\sum_{b=0}^{\infty} h(b) \left( \frac{1}{g_{i}} \right)^b\). According to the theory of power series convergence, the condition for this series to converge is that the exponential growth rate of \(h(b)\) does not exceed \(g_{i}\), i.e.,
		\[
		\limsup_{b \to \infty} h(b)^{1/b} < g_{i}.
		\]
		This implies that there exist constant \(\mathcal{C} > 0\) and   \(\alpha <g_i\) such that for all sufficiently large \(k\), \(h(k) \leq \mathcal{C} \alpha^k\).
		Letting  $\epsilon_i=\frac{\|L\|_1  m  }{c_{i}(g_{i} - l_i)} \sum_{b=0}^{\infty} h(b) \left( \frac{1}{g_{i}} \right)^b $, the  $\epsilon$-differential privacy is preserved with  $\epsilon=\max_{1\le i \le N}\epsilon_i$. The proof is completed.
	\end{proof}
	
	\
	\begin{remark}
		Compared to \cite{l1}, the two conclusions presented above relax the restrictions on \( h(k) \), eliminating the requirement for \( h(k) \) to decay exponentially. In fact, as shown in \eqref{hbb} of Conclusion \ref{con2}, we can set the magnitude of \( h(k) \) at any time step $k$ as long as the overall sum converges. Furthermore, \eqref{ca} in Conclusion \ref{con3} shows that we can arbitrarily set the magnitude of \( h(k) \) for finitely many terms, provided that for sufficiently large \( k \), the decay rate of \( h(k) \) is faster than \( g_i^k \).
	\end{remark}
	
	\subsection{Differentially private consensus}
	
	
	\begin{theorem}\label{DPC}
		Consider $b_i(k)=c_{i} p_{i}(k)$ with $c_{i}>0$ and $p_i(k)>0$. Let  the observer's initial values $\hat{x}_{y, i}(0)=\hat{x}_{y^{\prime}, i}(0)$ for adjacent  $y(k)$ and $y^{\prime}(k)$ with $m>0$ and  $h(k)>0$.
		Assume that there exist functions $h(k)$ and $p_{i}(k)$, and matrices $K$ and $L$ such that
		\begin{equation}\label{17}
			\begin{cases}
				\rho (A-L C)<1,\\
				\rho (I_{N-1} \otimes A-\Lambda \otimes B K)< 1,\\
				\sum_{k=0}^{\infty}p_{i}(k) < \infty, \\
				\epsilon=\max\limits_{1\le i \le N}\left(\|L\|_1  m \sum_{k=1}^{\infty }\frac{\sum_{b=0}^{k-1}l_i^{k-b-1}h(b)}{c_i p_i(k) }\right) < \infty,
			\end{cases}
		\end{equation}
		where $$\Lambda=\operatorname{diag}\left\{\lambda_{2}, \cdots, \lambda_{N}\right\}, \quad l_i=\left\|A-L C-d_i B K\right\|_1.$$
		Then, mean square consensus, almost sure consensus, and the preservation of $\epsilon$-differential privacy are simultaneously achieved.
	\end{theorem}
	
	\
	
	Then, we have the following corollary.
	
	\begin{corollary}\label{CDPP}
		Let  the observer's initial values $\hat{x}_{y, i}(0)=\hat{x}_{y^{\prime}, i}(0)$ for adjacent  $y(k)$ and $y^{\prime}(k)$ with $m>0$ and $h(k)=\alpha^k$,  $\alpha \in (0,1)$. Consider $b_i(k)=c_{i} g_{i}^{k}$, i.e., $p_i(k)=g_{i}^{k}$, $\alpha  <g_{i}<1$. Assume that there exist two matrices $K$ and $L$ such that $\rho(A-LC)<1$, $\rho (I_{N-1} \otimes A-\Lambda \otimes B K)< 1$, and $l_i=\|A-LC-d_iBK\|_1\in(\alpha,g_i)$  with $\Lambda=\operatorname{diag}\left\{\lambda_{2}, \cdots, \lambda_{N}\right\}$. 
		Then, mean square consensus, almost sure consensus, and the preservation of $\epsilon$-differential privacy are simultaneously achieved  with  $\epsilon=\max_{1\le i \le N}\epsilon_i$, where
		\begin{equation}\label{epi}
			\epsilon_i= \frac{m   g_{i}  \|L\|_1}{c_{i}\left(g_{i}-l_i\right)\left(g_{i}-\alpha\right)}.
		\end{equation}
	\end{corollary}
	\
	
	\begin{proof}
		Substituting  $h(k)=\alpha^k$ and $p_i(k)=g_{i}^{k}$ into equation \eqref{17}, we can obtain the above result.
	\end{proof}
	
	\

	\begin{remark}
		Corollary \ref{CDPP} tells us that it is possible to simultaneously achieve both consensus and differential privacy for the linear system using full-order state observer. This is different from previous work\cite{81}, which showed that achieving average consensus and differential privacy simultaneously is impossible for some first-order systems.	 { Moreover, compared with the conventional exponentially decaying noise adopted in much of the existing literature [20-24], our framework allows the noise  variance to decay at a {polynomial rate}, e.g., $b_i(k)=c_i/(k+1)^2$, while still guaranteeing $\epsilon$-differential privacy.}
	\end{remark}
	
	{
		\begin{remark}
			In the expression of the privacy parameter $\epsilon$, $m$ defines the maximum difference bound between adjacent datasets. A larger value of $m$ implies a more significant potential impact of data variation on the system output, thus requiring stronger noise to mask it. $\alpha$ controls the decay rate of the data difference itself, and a larger $\alpha$ similarly necessitates more persistent noise coverage. $\|L\|_1$, the 1-norm of the observer gain, reflects the extent to which the system amplifies output differences-the higher the gain, the more noise is required to compensate. $c_i$ represents the initial intensity of the noise, directly determining the initial level of privacy protection. The two key safety margins, $(g_i - l_i)$ and $(g_i - \alpha)$, ensure that the decay rate of the introduced noise is slower than the decay rate $l_i$ of the internal error dynamics and the decay rate $\alpha$ of the external data differences, respectively.   
	\end{remark}}
	
	{	\begin{remark}\label{app}
			Under the condition $l_i>\alpha$, we obtain the simplified approximate expression
			\[
			\epsilon \le \max_{1\le i \le N}\frac{m g_i |L|_1}{c_i (g_i-l_i)^2}.
			\]
			Due to the presence of an observer and the general noise design, this bound depends on more parameters. However, its fundamental structure remains consistent with those in [20] and [30]. This bound also provides clear guidelines for parameter selection: $\epsilon$ decreases with smaller $m$ and $\alpha$, and with larger $c_i$ and $g_i$. In particular, $\epsilon$ can be made arbitrarily small by increasing $c_i$.
		\end{remark}
	}
	
	{\begin{remark}
			When applying Corollary \ref{CDPP}, practitioners should first design the observer gain $L$ and controller gain $K$ such that $\rho(A-LC)<1$ and $\rho(I_{N-1} \otimes A-\Lambda \otimes B K)< 1$ are satisfied, based on the system matrix $A$ and the network eigenvalues matrix $\Lambda$.  The parameter $g_i$ should be chosen within $(\alpha, 1)$ and should also satisfy $g_i > l_i$, where $l_i = \|A - LC - d_i B K\|_1$. To achieve a desired level of privacy (a specific $\epsilon$), a trade-off exists: selecting a $g_i$ closer to $1$ or a smaller $c_i$ will result in stronger privacy protection (a smaller $\epsilon_i$) but will simultaneously slow down the convergence rate of the consensus algorithm. Therefore, $c_i$ and $g_i$ should be co-optimized to meet the application's required privacy budget $\epsilon$ while maintaining an acceptable consensus performance.
	\end{remark}}
	
	\subsection{ {$\epsilon^\star$-differentially private consensus}}
	{	In Corollary \ref{CDPP}, we established the conditions for the system to simultaneously achieve consensus and $\epsilon$-differential privacy, and provided the computational expression for $\epsilon$. If we want to specify a desired $\epsilon^\star$ value in advance, how should we configure the system parameters to achieve the intended level of privacy protection? This section is dedicated to addressing this question.}
	{	
		\begin{theorem}\label{CDP1}
			Let  the observer's initial values $\hat{x}_{y, i}(0)=\hat{x}_{y^{\prime}, i}(0)$ for adjacent $y(k)$ and $y^{\prime}(k)$ with $m>0$ and $h(k)=\alpha^k$,  $\alpha \in (0,1)$.
			Assume that there exist two matrices $K$ and $L$ such that $\rho(A-LC)<1$, $\rho (I_{N-1} \otimes A-\Lambda \otimes B K)< 1$, and $l_i=\|A-LC-d_iBK\|_1\in(\alpha,1)$  with $\Lambda=\operatorname{diag}\left\{\lambda_{2}, \cdots, \lambda_{N}\right\}$.  Consider $b_i(k)=c_{i} g_{i}^{k}$, i.e., $p_i(k)=g_{i}^{k}$, $g_{i}\in (l_i,1)$. 
			Then, for any desired privacy level $\epsilon^\star$, there exist   $g_{i}$ such that mean square consensus, almost sure consensus, and the preservation of $\epsilon^\star$-differential privacy are simultaneously achieved if the following condition holds:
			\begin{equation}\label{mL1}
				m \|L\|_1 < \epsilon^\star c_i (1 - \alpha)(1 - l_i).
			\end{equation}
		\end{theorem}
	}

\begin{proof}
{	Rewriting the original equation \eqref{epi} as a quadratic in \(g_i\):
		\[
		\epsilon^\star c_i \, g_i^2 - \big[ \epsilon^\star c_i (\alpha + l_i) + m \|L\|_1 \big] g_i + \epsilon^\star c_i \alpha l_i = 0.
		\]
		Then, the differentially private consensus problem is transformed into finding a solution \( g_i \in (l_i, 1) \) satisfying \( f(g_i) = 0 \).
		Define the quadratic function
		\[
		f(x) = \epsilon^\star c_i \, x^2 - \big[ \epsilon^\star c_i (\alpha + l_i) + m \|L\|_1 \big] x + \epsilon^\star c_i \alpha l_i.
		\]
		Since \(\epsilon^\star c_i > 0\), the graph of \(f(x)\) opens upward, and \(f(0) = \epsilon^\star c_i \alpha l_i > 0\).
		Evaluating the function at critical boundary points, we first examine the lower endpoint of the interval and we obtain:
		$$
		\begin{aligned}
			f(l_i) &= \epsilon^\star c_i l_i^2 - \left[ \epsilon^\star c_i (\alpha + l_i) + m \|L\|_1 \right] l_i + \epsilon^\star c_i \alpha l_i\\
			&= -m \|L\|_1 l_i < 0.
		\end{aligned}
		$$
		At the upper endpoint \(g = 1\), we have:
		\[
		f(1) = \epsilon^\star c_i (1 - \alpha)(1 - l_i) - m \|L\|_1.
		\]}
	
	{	We now demonstrate that the necessary and sufficient condition for the existence of \(g_i \in (l_i, 1)\) is \(m \|L\|_1 < \epsilon^\star c_i (1 - \alpha)(1 - l_i)\). For sufficiency, when this inequality holds, \(f(1) > 0\). Since \(f(l_i) < 0\) and \(f\) is continuous, the Intermediate Value Theorem ensures at least one real root in the open interval. To examine uniqueness, consider the derivative \(f'(x) = 2\epsilon^\star c_i x - \left[ \epsilon^\star c_i (\alpha + l_i) + m \|L\|_1 \right]\), whose zero occurs at \(x_0 = \frac{\epsilon^\star c_i (\alpha + l_i) + m \|L\|_1}{2\epsilon^\star c_i}\). If \(x_0 \leq l_i\), then \(f'(x) \geq 0\) on the interval, guaranteeing monotonic increase and a unique root. If \(x_0 > l_i\), then the combination of negative lower boundary value and upward-opening parabolic shape still yields a unique root. For necessity, if \(m \|L\|_1 \geq \epsilon^\star c_i (1 - \alpha)(1 - l_i)\), then \(f(1) \leq 0\). Together with \(f(l_i) < 0\) and the upward-opening nature, this implies \(f(x) < 0\) throughout the interval when \(f(1) < 0\), or vanishing only at the endpoint \(x = 1\) when \(f(1) = 0\). Since the endpoint is excluded, no valid solution exists in the open interval.}
	
	{		In conclusion, the necessary and sufficient condition for the existence of a unique \(g_i \in (l_i, 1)\) is:
		\[
		m \|L\|_1 < \epsilon^\star c_i (1 - \alpha)(1 - l_i).
		\]
		This completes the proof.}
\end{proof}

\
	
	{	
		\begin{remark}
			\label{rmk:design_freedom}
			The parameters $m$ and $c_i$ in condition (\ref{mL1}) offer complementary design freedom to fulfill the privacy and consensus requirements.
			Crucially, for any predefined privacy level $\epsilon^\star > 0$, one can always ensure the condition $m \|L\|_1 < \epsilon^\star c_i (1 - \alpha)(1 - l_i)$ holds for all agents $i$ by either:
			\begin{itemize}
				\item Employing a noise sequence with a sufficiently large initial amplitude $c_i$, or
				\item For a given noise sequence (fixed $c_i$), ensuring the initial deviation $m$ between adjacent datasets is sufficiently small.
			\end{itemize}
			This flexibility provides a clear design pathway: by appropriately scaling the initial noise intensity $c_i$ relative to the system's inherent sensitivity $m$, one can simultaneously achieve the desired $\epsilon^\star$-differential privacy guarantee while maintaining consensus.
	\end{remark}}
	{	\begin{remark}
			For the proposed observer-based differentially private consensus algorithm, we have analyzed the computational complexity of each component. For each agent, one iteration mainly consists of the observer update with complexity $O(n^{2})$, the distributed control computation with complexity $O(d_i n)$, and the generation of the Laplace noise vector with complexity $O(n)$. Therefore, the overall per-agent computational complexity is $O(n^{2})$. Since all computations and communications are performed locally, the algorithm exhibits good scalability. The total computational burden of the entire network is $O(N n^{2})$, which does not increase significantly with the number of agents $N$.
	\end{remark}}

	\section{ Differentially private consensus based on the reduced-order observer}\label{Sreduce}
	
	Compared with the full-order observer, the reduced-order observer can reduce the dimension of the observer system and simplify the computation. Without loss of generality,  the original system \eqref{mainagent} is assumed to be the following form
	\begin{equation}\label{resys1}
		\begin{aligned}
			& \left\{\begin{aligned}
				\left[\begin{array}{l}
					\overline{x}_{1 i}(k+1)\\
					\overline{x}_{2 i}(k+1)
				\end{array}\right] =&\left[\begin{array}{ll}
					\overline{A}_{11} &  \overline{A}_{12}\\
					\overline{A}_{21} & \overline{A}_{22}
				\end{array}\right]\left[\begin{array}{l}
					\overline{x}_{1 i}(k) \\
					\overline{x}_{2 i}(k)
				\end{array}\right]\\
				& +\left[\begin{array}{l}
					\overline{B}_1 \\
					\overline{B}_2
				\end{array}\right] u_i(k), \\
				y_i(k) =& \left[\begin{array}{ll}
					0 & I_q
				\end{array}\right]\left[\begin{array}{l}
					\overline{x}_{1i}(k) \\
					\overline{x}_{2i}(k)
				\end{array}\right],
			\end{aligned}\right.
		\end{aligned}
	\end{equation}
	where $\overline{x}_{1i}(k) \in \mathbb{R}^{n-q}$, $\overline{x}_{2i}(k) \in \mathbb{R}^q$, $\overline{A}_{11}\in \mathbb{R}^{(n-q) \times (n-q)}$, $\overline{A}_{12}\in \mathbb{R}^{(n-q) \times q}$, $\overline{A}_{21}\in \mathbb{R}^{q \times (n-q)}$, $\overline{A}_{22}\in \mathbb{R}^{q \times q}$, $\overline{B}_{1}\in \mathbb{R}^{(n-q) \times r}$, and $\overline{B}_{2}\in \mathbb{R}^{q \times r}$. This is equivalent to
	\begin{equation}\label{resys2}
		\begin{aligned}
			& \left\{\begin{array}{l}
				\overline{x}_{1 i}(k+1)=\overline{A}_{11} \overline{x}_{1 i}(k)+\overline{A}_{12} \overline{x}_{2 i}(k)+\overline{B}_1 u_i(k) \\
				\overline{x}_{2 i}(k+1)=\overline{A}_{21} \overline{x}_{1 i}(k)+\overline{A}_{22} \overline{x}_{2 i}(k)+\overline{B}_2 u_i(k) \\
				y_i(k)=\overline{x}_{2i}(k).
			\end{array}\right.
		\end{aligned}
	\end{equation}
	
	Letting $\overline{u}(k)=\overline{A}_{12} \overline{x}_{2 i}(k)+\overline{B}_1 u_i(k)$ and $\overline{y}(k)=\overline{A}_{21} \overline{x}_{1i}(k)=\overline{x}_{2 i}(k+1)-\overline{A}_{22} \overline{x}_{2 i}(k)-\overline{B}_2 u_i(k)$, we have
	$$
	\left\{\begin{array}{l}
		\overline{x}_{1 i}(k+1)=\overline{A}_{11} \overline{x}_{1 i}(k)+\overline{u}(k) \\
		\overline{y}(k)=\overline{A}_{21} \overline{x}_{1 i}(k).
	\end{array}\right.
	$$
	The reduced-order observer is given as follows
	\begin{equation}\label{obser2}
		\begin{aligned}
			\left\{\begin{array}{l}\hat{\overline{x}}_{1 i}(k+1)=\overline{A}_{11} \hat{\overline{x}}_{1 i}(k)+\overline{u}(k)+\overline{L} \left(\overline{y}(k)-\hat{\overline{y}}(k)\right), \\
				\hat{\overline{y}}(k)=\overline{A}_{21} \hat{\overline{x}}_{1 i}(k) ,\end{array}\right.
		\end{aligned}
	\end{equation}
	where $ \overline{L} \in \mathbb{R}^{n \times q}$ is the state estimation gain.
	
	Letting $\hat{x}_i(k)=\left[\begin{array}{l}\hat{\overline{x}}_{1 i}(k) \\ \overline{x}_{2 i}(k)\end{array}\right]$, the information that the Agent $i$ sent to other agents is generated as
	\begin{equation}\label{addlap2}
		{\theta}_{i}(k)=\hat{x}_i(k)+\eta_{i}(k)=\left[\begin{array}{l}\hat{\overline{x}}_{1 i}(k) \\ \overline{x}_{2 i}(k)\end{array}\right]+\eta_{i}(k),
	\end{equation}
	where $\eta_{i}(k) \sim \operatorname{Lap}(0, b_i(k))$, and the observer-based controller can be rewritten as follows
	\begin{equation}\label{u2}
		u_i(k)=K \sum_{j \in \mathcal{N}_i} a_{i j}\left({\theta}_j(k)-\left[\begin{array}{l}\hat{\overline{x}}_{1 i}(k) \\ \overline{x}_{2 i}(k)\end{array}\right]\right),
	\end{equation}
	where $K=\left[\begin{array}{ll}K_1 & K_2\end{array}\right]$, $K_1 \in \mathbb{R}^{r \times (n-q)}$, $K_2 \in \mathbb{R}^{r \times q}$.
	
	\
	
	Next, we will derive the sufficient conditions for achieving consensus and differential privacy based on the observer \eqref{obser2}, the noise addition mechanism \eqref{addlap2}, and the controller \eqref{u2}.
	\subsection{Consensus}
	Similar to the case of  full-order observer in Section \ref{Sfull}, we first give the following  mean square and almost sure consensus theorem.
	\begin{theorem}\label{consen2}
		Consider  $b_i(k)=c_{i} p_{i}(k)$ with $c_{i}>0$, $p_{i}(k)>0$, and $\sum_{k=0}^{\infty}p_{i}(k) < \infty$. If  there exist matrices $K$ and $\overline{L}$ such that $\rho(I_{N-1} \otimes A-\Lambda \otimes B K)<1$ and $\rho(\overline{A}_{11}-\overline{L} \overline{A}_{21})<1$ with $\Lambda=\operatorname{diag}\left\{\lambda_{2}, \cdots, \lambda_{N}\right\}$, then both the mean  square consensus and the  almost sure  consensus are achieved for  the MAS \eqref{resys1} with the observer \eqref{obser2} and the control \eqref{u2}.
	\end{theorem}
	\begin{proof}
		Let $\overline{e}_{1 i}(k)=\overline{x}_{1 i}(k)-\hat{\overline{x}}_{1 i}(k)$. Then, we have from \eqref{resys1} and \eqref{obser2} that
		$$
		\overline{e}_{1i}(k+1)=(\overline{A}_{11}-\overline{L}\overline{A}_{21}) \overline{e}_{1i}(k).
		$$
		Letting $\overline{e}_1(k)=\left[\overline{e}_{11}^{\mathrm{T}}(k), \overline{e}_{12}^{\mathrm{T}}(k), \cdots, \overline{e}_{1N}^{\mathrm{T}}(k)\right]$, one can find
		$$
		\overline{e}_1(k+1)=\left[I_N \otimes\left(\overline{A}_{11}-\overline{L} \overline{A}_{21}\right)\right] \overline{e}_1(k).
		$$
		If there exists matrix $\overline{L}$ such that $\rho(\overline{A}_{11}-\overline{L} \overline{A}_{21})<1$, then it follows that $$\lim\limits_{k  \rightarrow \infty} \overline{e}_1(k)=0.$$
		
		From \eqref{mainagent}, \eqref{addlap2}, and \eqref{u2}, we have
		$$
		\begin{aligned}
			x_i(k&+1)= A x_i(k)+B K \sum_{j \in \mathcal{N}_i} a_{i j}\left({\theta}_j(k)-\left[\begin{array}{c}
				\hat{\overline{x}}_{1 i}(k) \\
				\overline{x}_{2 i}(k)
			\end{array}\right]\right) \\
			=& A x_i(k)-B K \sum_{j=1}^N\left[L_G\right]_{i j} x_j(k)\\
			&+B K_1 \sum_{j=1}^N\left[L_G\right]_{i j} \overline{e}_{1 j}(k)-B K  \sum_{j=1, j \neq i}^N\left[L_G\right]_{i j} \eta_{j}(k).
		\end{aligned}
		$$
		We can rewrite the above equation in a compact form as follows
		\begin{equation}\label{comp2}
			\begin{aligned}
				x(k+1)=&\left[\left(I_N \otimes A\right)-\left(L_G \otimes B K\right)\right] x(k)\\
				&+\left(L_G \otimes B K_1\right) \overline{e}_{1}(k)+\left(A_G \otimes B K\right) \eta(k).
			\end{aligned}
		\end{equation}
		It can be seen that \eqref{comp2} has a same form of \eqref{comp1}.  If  there exist matrices $K$ and $\overline{L}$ such that $\rho(I_N \otimes A-\Lambda \otimes B K)<1$ and $\rho(\overline{A}_{11}-\overline{L} \overline{A}_{21})<1$, then both mean  square and almost sure consensus are achieved for the system \eqref{mainagent} with the observer \eqref{obser2} by applying similar procedures to Theorem \ref{consen1}.
	\end{proof}
	
	\
	
	The following conclusion immediately follows.
	\begin{corollary}\label{coro6}
		Consider  $b_i(k)=c_{i} g_{i}^{k}$ with $c_{i}>0$ and $0<g_{i}<1$, i.e., $p_i(k)=g_{i}^{k}$. If  there exist matrices $K$ and $\overline{L}$ such that $\rho(I_N \otimes A-\Lambda \otimes B K)<1$ and $\rho(\overline{A}_{11}-\overline{L} \overline{A}_{21})<1$ with $\Lambda=\operatorname{diag}\left\{\lambda_{2}, \cdots, \lambda_{N}\right\}$, then both the mean  square consensus and the  almost sure  consensus are achieved for  the MAS \eqref{resys1} with the observer \eqref{obser2} and the control \eqref{u2}.
	\end{corollary}
	
	\begin{proof}
		It is easy to prove that $p_i(k)=g_{i}^{k}$ with $0<g_{i}<1$ satisfies conditions $p_{i}(k)>0$ and $\sum_{k=0}^{\infty}p_{i}(k) < \infty$, so we obtain the above corollary.
	\end{proof}
	
	\
	
	Following the similar methods as in Corollary \ref{corostab}, we can derive the following observer-based stabilization result for a single linear system ($N=1$) with
	\begin{equation}\label{addlap251}
		u(t)=K(\hat{x}(k)+\eta(k))=K\left[\begin{array}{l}\hat{\overline{x}}_{1}(k) \\ \overline{x}_{2}(k)\end{array}\right]+\eta(k).
	\end{equation}
	\begin{corollary}
		Consider  $b(k)=c p^{k}$ with $c>0$ and $0<p<1$. If  there exist matrices $K$ and $\overline{L}$ such that $\rho( A-  B K)<1$ and $\rho(\overline{A}_{11}-\overline{L} \overline{A}_{21})<1$,  then both the mean  square stabilization and the  almost sure  stabilization are achieved for  the single system with the observer \eqref{obser2} and the control \eqref{addlap251}.
	\end{corollary}
	
	
	\
	
	{The following theorem explicitly characterizes the mean-square convergence rate for the case of exponentially decaying scale parameter $b_i(k)$. }
	{	\begin{theorem}
			Consider $b_i(k)=c_{i} g_{i}^{k}$ with $c_{i}>0$ and $0<g_{i}<1$. Assume that there exist two matrices $K$ and $\overline{L}$ such that $\rho(I_N \otimes A-\Lambda \otimes B K)<1$ and $\rho(\overline{A}_{11}-\overline{L} \overline{A}_{21})<1$ with $\Lambda=\operatorname{diag}\left\{\lambda_{2}, \cdots, \lambda_{N}\right\}$. Then, the mean square convergence rate is
			$$
			\begin{aligned}
				\rho = \max\left( \rho(I_N \otimes A - \Lambda \otimes BK), \rho(\overline{A}_{11}-\overline{L}\overline{A}_{21}), \max_i g_i \right).
			\end{aligned}
			$$
	\end{theorem}}
	
	\begin{proof}
		The proof is analogous to that of Theorem \ref{mscr1} and is  therefore omitted here.
	\end{proof}
	
	\
	
	{	\begin{remark}
			This theorem establishes the mean-square convergence rate under a reduced-order observer, exhibiting a clear correspondence with the full-order observer case. Compared with Theorem \ref{mscr1}, the convergence rate is similarly determined by three factors: the cooperative control dynamics $\rho(I_N \otimes A - \Lambda \otimes BK)$, the observer error dynamics $\rho(\overline{A}_{11}-\overline{L}\overline{A}_{21})$, and the noise decay rate $\max_i g_i$. The key distinction lies in the observer dynamics term-while the full-order observer considers the spectral radius of the full estimation error system matrix $A-LC$, the reduced-order observer only involves the spectral radius of the reduced subsystem matrix $\overline{A}_{11}-\overline{L}\overline{A}_{21}$. This difference reflects the core advantage of reduced-order observers: by leveraging structural information of the output matrix, the stability condition for observer design is reduced from the full-dimensional space to a lower-dimensional subspace, potentially yielding faster convergence characteristics. However, both cases clearly demonstrate that the decay rate of the privacy-preserving noise directly constrains the ultimate system convergence performance, establishing a unified theoretical framework for balancing privacy and convergence rate under different observer structures.
	\end{remark}}
	
	\subsection{Differential privacy}\label{S43}
	The next result gives some sufficient conditions to preserve the differential privacy based on the reduced-order observer \eqref{obser2}. 
	\begin{theorem}\label{DP2}
		Consider $b_i(k)=c_{i} p_{i}(k)$ with $c_i>0$ and $p_i(k)>0$.
		Let the observer's initial values $\hat{\overline{x}}_{y, 1 i}(0)=\hat{\overline{x}}_{y^{\prime}, 1 i}(0)$ for adjacent $y(k)$ and $y^{\prime}(k)$ with $m>0$ and $h(k)>0$.  If there exist  functions $h(k)$ and $p_{i}(k)$, and matrices $K_1$ and $K_2$ such that
		\begin{equation}\label{}
			\begin{aligned}
				\epsilon_i= m\sum_{k=1}^{\infty} \frac{w_{i} \sum_{l=0}^{k-1} v_{i}^{k-l-1} h(l)}{c_{i}p_{i}(k)}+m \sum_{k=0}^{\infty}\frac{h(k)}{c_{i}p_{i}(k)} < \infty,
			\end{aligned}
		\end{equation}
		then the  $\epsilon$-differential privacy is preserved  with  $\epsilon=\max_{1\le i \le N}\epsilon_i$, where $v_{i}=\left\|\overline{A}_{11}-d_i \overline{B}_1 K_1\right\|_1$ and $w_i=\left\|\overline{A}_{12}-d_i \overline{B}_1 K_2\right\|_1$.
	\end{theorem}
	
	\begin{proof}
		From \eqref{resys2} and Definition \ref{adjacent}, we can see that $y(k)$ and $y^{\prime}(k)$ are adjacent if, there exist $i_0 \in \{1,2,...,N \}$ and $k_{0} \geq 0$ such that for $i \neq i_0$, $y_i(k)=y_i^{\prime}(k), k \in \mathbb{Z}_{\geq 0}$, and
		\begin{equation}\label{adj2}
			\begin{aligned}
				\begin{cases}\overline{x}_{2i_0}(k)=\overline{x}_{2i_0}^{\prime}(k), & k<k_{0} \\ \left\|\overline{x}_{2i_0}(k)-\overline{x}_{2i_0}^{\prime}(k)\right\|_1 \leq m h(k-k_{0}), & k \geq k_{0}.\end{cases}
			\end{aligned}
		\end{equation}
		
		Based on the reduced-order observer \eqref{obser2}, define function $\overline{\varphi}(\eta, y)={\theta}$ and  an observation $\bar{O}=\left\{{\theta}_{i}(k), i\in \mathcal{V}\right\}$. It is obvious that
		$$
		\begin{aligned}
			P\{\overline{\varphi}(\eta(k), y(k)) \in \bar{O}\}&=P\{\hat{x}_{y, i}(k)+\eta_{i}(k) \in \bar{O}\}\\
			&=P\{\eta_{i}(k) \in \bar{O}-\hat{x}_{y,i}(k)\}.
		\end{aligned}
		$$
		Letting
		$$
		\begin{aligned}
			\overline{\beta}_{i}(k)&=\hat{{x}}_{y, i}(k)-\hat{{x}}_{y^\prime, i}(k)
			&=\left[\begin{array}{l}
				\hat{\overline{x}}_{y, 1 i}(k)-\hat{\overline{x}}_{y^{\prime}, 1 i}(k) \\
				\overline{x}_{y, 2 i}(k)-\overline{x}_{y^{\prime}, 2 i}(k)
			\end{array}\right],
		\end{aligned}
		$$
		we have
		$$
		\begin{aligned}
			P\{\overline{\varphi}(\eta(k), y^{\prime}(k)) \in \bar{O}\}&=P\{\hat{x}_{y^\prime, i}(k)+\eta_{i}(k) \in \bar{O}\}\\
			&=P\{\hat{x}_{y, i}(k)- \overline{\beta}  _{i}(k)+\eta_{i}(k) \in \bar{O}\}\\
			&=P(\eta_{i}(k)- \overline{\beta}  _{i}(k) \in \bar{O}-\hat{x}_{y,i}(k)\}.
		\end{aligned}
		$$
		Thus, we obtain
		$$
		\begin{aligned}
			&\frac{\mathbb{P}\left\{\boldsymbol{\eta}(k) \in \Omega \mid \overline{\varphi} (\boldsymbol{\eta}(k), y(k)) \in \bar{O}\right\}}{ \mathbb{P}\left\{\boldsymbol{\eta}(k) \in \Omega \mid \overline{\varphi} (\boldsymbol{\eta}(k), y^{\prime}(k)) \in \bar{O}\right\}} \\
			&\quad =\prod_{i\in \mathcal{V}}\prod_{k=0}^{\infty}\frac{\int_{\bar{O}-\hat{x}_{y,i}(k)}L(\eta_i(k),0,b_i(k))d\eta_i(k)}{\int_{\bar{O}-\hat{x}_{y,i}(k)}L(\eta_i(k),- \overline{\beta}  _i(k),b_i(k))d\eta_i(k)}   \\
			& \quad=\prod_{i\in \mathcal{V}}\prod_{k=0}^{\infty}e^{ -\frac{\|\eta_i(k) \|_1}{b_i(k)}+\frac{\| \eta_i(k)+ \overline{\beta}  _i(k)\|_1}{b_i(k)}  }\\
			&\quad \leq \prod_{i\in \mathcal{V}}\prod_{k=0}^{\infty}e^{ \frac{\| \overline{\beta}  _i(k) \|_1}{b_i(k)} }\\
			& \quad= e^{\sum_{i\in \mathcal{V}}\sum_{k=0}^{\infty}\frac{\| \overline{\beta}  _i(k) \|_1}{b_i(k)}}.
		\end{aligned}
		$$
		
		Next, we aim to calculate $\| \overline{\beta}  _i(k) \|_1$.  From \eqref{obser2} and \eqref{u2}, it follows that
		$$
		\begin{aligned}
			\hat{\overline{x}}_{1 i}(k+&1)=\overline{A}_{11} \hat{\overline{x}}_{1 i}(k)+\overline{A}_{12} \overline{x}_{2 i}(k)+\overline{L}(\overline{y}-\hat{\overline{y}})\\
			&+\overline{B}_1\left[K \sum_{j \in \mathcal{N}_i} a_{i j}\left({\theta}_j(k)-\left[\begin{array}{c}
				\hat{\overline{x}}_{1 i}(k) \\
				\overline{x}_{2 i}(k)
			\end{array}\right]\right)\right] \\
			=&\left(\overline{A}_{11}-\overline{L} \overline{A}_{21}\right) \hat{\overline{x}}_{1i}(k)+\overline{L} \overline{A}_{21} \overline{x}_{1i}(k)+\overline{A}_{12} \overline{x}_{2 i}(k) \\
			& -c_i \overline{B}_1\left(K_1 \hat{\overline{x}}_{1i}(k)+K_2 \overline{x}_{2 i}(k)\right)\hspace{-0.08cm}+\hspace{-0.08cm}\overline{B}_1 K \sum_{j \in \mathcal{N}_i} a_{i j} \overline{\theta_j}(k) \\
			=&\left(\overline{A}_{11}-t_i \overline{B}_1 K_1\right) \hat{\overline{x}}_{1i}(k)+\overline{L}\overline{A}_{21} \overline{e}_{1 i}(k)\\
			&+\left(\overline{A}_{12}-t_i \overline{B}_1 K_2\right) \overline{x}_{2 i}(k)+\overline{B}_1 K \sum_{j \in \mathcal{N}_i} a_{i j} {\theta}_j(k).
		\end{aligned}
		$$
		Since $\bar{O}=\left\{{\theta}_i(k), i\in \mathcal{V}\right\}$ for $y$ and $y^\prime$ are the some, we have
		$$
		\begin{aligned}
			\hat{\overline{x}}_{y, 1 i}(k+1)-&\hat{\overline{x}}_{y^{\prime}, 1 i}(k+1)\\
			= & \left(\overline{A}_{11}-t_i \overline{B}_1 K_1\right)\left[\hat{\overline{x}}_{y, 1 i}(k)-\hat{\overline{x}}_{y^{\prime}, 1 i}(k)\right] \\
			&+\overline{L} \overline{A}_{21}(\overline{e}_{y, 1 i}(k)-\overline{e}_{y^{\prime}, 1 i}(k)) \\
			&+\left(\overline{A}_{12}-t_i \overline{B}_1 K_2\right)\left[\overline{x}_{y, 2 i}(k)-\overline{x}_{y^{\prime}, 2 i}(k)\right].
		\end{aligned}
		$$
		Letting $\hat{\overline{x}}_{y, 1 i}(0)=\hat{\overline{x}}_{y^{\prime}, 1 i}(0)$, we obtain $$\overline{e}_{y, 1 i}(0)-\overline{e}_{y^\prime, 1 i}(0)=0$$ and
		$$
		\begin{aligned}
			\hat{\overline{x}}_{y, 1 i}(k&+1)- \hat{\overline{x}}_{y^{\prime}, 1 i}(k+1)\\
			=&\sum_{b=0}^k\left(\overline{A}_{11}-t_i \overline{B}_1 K_1\right)^{k-b}\{[\overline{A}_{2 1}\left(\overline{e}_{y, 1 i}(b)-\overline{e}_{y^{\prime}, 1 i}(b)\right)\\
			&+\left(\overline{A}_{12}-t_i \overline{B}_1 K_2\right)\left[\overline{x}_{y, 2 i}(b)-\overline{x}_{y^{\prime}, 2 i}(b)\right]\}\\
			=&\sum_{b=0}^k\left(\overline{A}_{11}-t_i \overline{B}_1 K_1\right)^{k-b} \\
			&  \times\left\{\overline{L} \overline{A}_{21}\left(\overline{A}_{11}-\overline{L} \overline{A}_{21}\right)^b \left[\overline{e}_{y, 1 i}(0)-\overline{e}_{y^\prime, 1 i}(0)\right]\right.\\
			&\quad+\left. \left(\overline{A}_{12}-t_i \overline{B}_1 K_2\right)\left[\overline{x}_{y, 2 i}(b)-\overline{x}_{y^\prime, 2 i}(b)\right] \right\} \\
			=&\sum_{b=0}^k\left(\overline{A}_{11}-t_i \overline{B}_1 K_1\right)^{k-b} \\
			&  \times \left(\overline{A}_{12}-t_i \overline{B}_1 K_2\right)\left[\overline{x}_{y, 2 i}(b)-\overline{x}_{y^\prime, 2 i}(b)\right]\\
			=&\sum_{b=0}^k v_i^{k-b} w_i \left[\overline{x}_{y, 2 i}(b)-\overline{x}_{y^\prime, 2 i}(b)\right].
		\end{aligned}
		$$
		For $i=i_0$, we have
		$$
		\begin{aligned}
			\|\hat{\overline{x}}_{y, 1 i_0}(k&+1)-\hat{\overline{x}}_{y^{\prime},1 i_0}(k+1)\|_1 \\
			=&\sum_{b=0}^k v_{i_0}^{k-b}w_{i_0} \|\overline{x}_{y, 2 i_0}(b)-\overline{x}_{y^\prime, 2 i_0}(b)\|_1 \\
			& \leq w_{i_0} \sum_{b=k_0}^k v_{i_0}^{k-b}  m h(b-k_0)\\
			& \leq m w_{i_0}  \sum_{b=0}^k v_{i_0}^{k-b} h(b) .\\
		\end{aligned}
		$$
		This together with \eqref{adj2} implies
		$$
		\begin{aligned}
			\|\overline{\beta}_{i_0}(k)\|_1 \hspace{-0.08cm}&\leq \hspace{-0.08cm}\left\|\hat{\overline{x}}_{y, 1 i_0}(k)\hspace{-0.08cm}-\hspace{-0.08cm}\hat{\overline{x}}_{y^\prime,1 i_0}(k)\right\|_1\hspace{-0.08cm}+\hspace{-0.08cm}\| \overline{x}_{y, 2 i_0}(k)\hspace{-0.08cm}-\hspace{-0.08cm}\overline{x}_{y^\prime, 2 i_0}(k) \|_1\\
			& \leq m w_{i_0}  \sum_{l=0}^{k-1} v_{i_0}^{k-l-1} h(l) + m h(k).
		\end{aligned}
		$$
		For $i \ne i_0$, we have $$\|\hat{\overline{x}}_{y, 1 i}(k+1)-\hat{\overline{x}}_{y,1 i}(k+1)\|_1=0 $$ and $$ \| \overline{\beta}  _{i}(k) \|_1=0.$$
		Then, we have
		$$
		\begin{aligned}
			&\frac{\mathbb{P}\left\{\boldsymbol{\eta}(k) \in \Omega \mid \overline{\varphi} (\boldsymbol{\eta}(k), y(k)) \in \bar{O}\right\}}{ \mathbb{P}\left\{\boldsymbol{\eta}(k) \in \Omega \mid \overline{\varphi} (\boldsymbol{\eta}(k), y^{\prime}(k)) \in \bar{O}\right\}}\\
			& \quad\leq e^{m  \sum_{k=1}^{\infty} \frac{w_{i_0} \sum_{l=0}^{k-1} v_{i_0}^{k-l-1} h(l)}{c_{i_0}p_{i_0}(k)}+m \sum_{k=0}^{\infty}\frac{h(k)}{c_{i_0}p_{i_0}(k)}} < \infty.
		\end{aligned}
		$$
		Letting  $\epsilon_i=m  \sum_{k=1}^{\infty} \frac{w_{i} \sum_{l=0}^{k-1} v_{i}^{k-l-1} h(l)}{c_{i}p_{i}(k)}+m \sum_{k=0}^{\infty}\frac{h(k)}{c_{i}p_{i}(k)}$, the  $\epsilon$-differential privacy can be preserved with $\epsilon=\max_{1\le i \le N}\epsilon_i$. The proof is completed.
	\end{proof}
	
	\
	
	Then, we have the following corollaries. 
	The proof follows similar procedures to those in Corollaries \ref{con2} and \ref{con3} and is therefore omitted here.
	
	\begin{corollary}\label{con4}
		Consider  $b_i(k)=c_{i}\frac{1}{(k+1)^2}$, i.e., $p_i(k)=\frac{1}{(k+1)^2}$. Let the observer's initial values $\hat{x}_{y, i}(0)=\hat{x}_{y^{\prime}, i}(0)$ for adjacent $y(k)$ and $y^{\prime}(k)$ with $m>0$ ,  $h(k)>0$, and
		\begin{equation}\label{hbb2}
			\sum_{k=0}^{\infty}h(k)k^2 < \infty.
		\end{equation}
		If there exist function $h(k)$ and matrix $K_1$ such that $v_i<1$, then the  $\epsilon$-differential privacy is preserved  with  $\epsilon=\max_{1\le i \le N}\epsilon_i$, where $\epsilon_i= m w_i {c_i}^{-1}(1 - v_{i})^{-3}\sum_{b=0}^{\infty} h(b) \left[ (b + 2)^2 - (2b^2 + 6b + 3) v_{i} + (b + 1)^2 v_{i}^2 \right]+m{c_i}^{-1}\sum_{b=0}^{\infty}h(b)(b+1)^2$.
	\end{corollary}

	\begin{corollary}\label{con5}
		Consider  $b_i(k)=c_{i} g_{i}^{k}$, i.e., $p_i(k)=g_{i}^{k}$,  $0<g_i<1$. Let the observer's initial values $\hat{x}_{y, i}(0)=\hat{x}_{y^{\prime}, i}(0)$ for adjacent $y(k)$ and $y^{\prime}(k)$ with $m>0$ and  $h(k)>0$. Assume that there exist constant $\mathcal{C}$ and real number $\alpha < g_i$ such that
		\begin{equation}\label{ca2}
			h(k) \leq \mathcal{C} \alpha^k
		\end{equation}
		for all sufficiently large $k$. If there exist matrices $K$ and $L$ such that $v_{i}\in(\alpha, g_i)$, then the  $\epsilon$-differential privacy is preserved  with  $\epsilon=\max_{1\le i \le N}\epsilon_i$, where $\epsilon_i=\frac{m (w_i+g_i-v_i)  }{c_{i}(g_{i} - v_{i})} \sum_{b=0}^{\infty} h(b) \left( \frac{1}{g_{i}} \right)^b $.
	\end{corollary}
	
	\subsection{Differentially private consensus}

	\begin{theorem}\label{}
		Consider   $b_i(k)=c_{i} p_{i}(k)$ with $c_{i}>0$ and $p_i(k)>0$. Let  the observer's initial values $\hat{x}_{y, i}(0)=\hat{x}_{y^{\prime}, i}(0)$ for adjacent $y(k)$ and $y^{\prime}(k)$ with $m>0$ and  $h(k)>0$.
		Assume that there exist functions $h(k)$ and $p_{i}(k)$ and matrices $K$ and $\overline{L}$ such that
		\begin{equation}\label{}
			\begin{cases}
				\rho(\overline{A}_{11}-\overline{L} \overline{A}_{21})<1,\\
				\rho(I_{N-1} \otimes A-\Lambda \otimes B K)<1,\\
				\sum_{k=0}^{\infty}p_{i}(k) < \infty, \\
				\sum_{k=1}^{\infty} \frac{w_{i} \sum_{l=0}^{k-1} v_{i}^{k-l-1} h(l)}{c_{i}p_{i}(k)}+ \sum_{k=0}^{\infty}\frac{h(k)}{c_{i}p_{i}(k)} < \infty
			\end{cases}
		\end{equation}
		with $\Lambda=\operatorname{diag}\left\{\lambda_{2}, \cdots, \lambda_{N}\right\}$. Then, mean square consensus, almost sure consensus, and the preservation of $\epsilon$-differential privacy are simultaneously achieved  with  $\epsilon=\max_{1\le i \le N}\epsilon_i$, where
		$$\epsilon_i= m\sum_{k=1}^{\infty} \frac{w_{i} \sum_{l=0}^{k-1} v_{i}^{k-l-1} h(l)}{c_{i}p_{i}(k)}+m \sum_{k=0}^{\infty}\frac{h(k)}{c_{i}p_{i}(k)}.$$
	\end{theorem}
	
	\
	
	Similarly to Corollary \ref{CDP1} based on the full-order state observer, we can obtain the joint design for achieving consensus and preservation of  $\epsilon$-differential privacy and the proof is omitted here.
	\begin{corollary}\label{CDPPP}
		Let  the observer's initial values $\hat{\overline{x}}_{y, 1 i}(0)=\hat{\overline{x}}_{y^{\prime}, 1 i}(0)$ for adjacent $y(k)$ and $y^{\prime}(k)$ with $m>0$ and $h(k)=\alpha^k$, $\alpha \in (0,v_i)$. Consider $b_i(k)=c_{i} g_{i}^{k}$ with $v_i<g_{i}<1$.  Assume there exist matrices $K$ and $\overline{L}$ such that $ \rho(I_{N-1} \otimes A-\Lambda \otimes B K)<1$ and $ \rho(\overline{A}_{11}-\overline{L} \overline{A}_{21})<1$.
		Then, mean square consensus, almost sure consensus, and the preservation of the $\epsilon$-differential privacy are  achieved  with  $\epsilon=\max_{1\le i \le N}\epsilon_i$, where $\epsilon_i=\frac{ m  g_{i} (w_{i}+g_{i}-v_{i})}{c_{i}\left(g_{i}-v_{i}\right)\left(g_{i}-\alpha\right)}$,  $v_{i}=\left\|\overline{A}_{11}-d_i \overline{B}_1 K_1\right\|_1$, and $w_i=\left\|\overline{A}_{12}-d_i \overline{B}_1 K_2\right\|_1$.
	\end{corollary}
	
	
	\subsection{ {$\epsilon^\star$-differentially private consensus}}
{	\begin{theorem}\label{CDP2}
		Let  the observer's initial values $\hat{\overline{x}}_{y, 1 i}(0)=\hat{\overline{x}}_{y^{\prime}, 1 i}(0)$ for adjacent $y(k)$ and $y^{\prime}(k)$ with $m>0$ and $h(k)=\alpha^k$, $\alpha \in (0,v_i)$. Consider $b_i(k)=c_{i} g_{i}^{k}$, $g_{i}\in (v_i,1)$.  Assume there exist matrices $K$ and $\overline{L}$ such that $ \rho(I_{N-1} \otimes A-\Lambda \otimes B K)<1$ and $ \rho(\overline{A}_{11}-\overline{L} \overline{A}_{21})<1$.
		Then, for any desired privacy level $\epsilon^\star > 0$,  there exist  $g_{i}$ such that mean square consensus, almost sure consensus, and $\epsilon^\star$-differential privacy are simultaneously achieved if the following condition holds:
		\begin{equation}
			m(w_i + 1 - v_i) < \epsilon^\star c_i (1 - \alpha)(1 - v_i),
		\end{equation}
		where $v_{i}=\left\|\overline{A}_{11}-d_i \overline{B}_1 K_1\right\|_1$ and $w_i=\left\|\overline{A}_{12}-d_i \overline{B}_1 K_2\right\|_1$.
\end{theorem}}

	\begin{proof}
	{	The privacy preservation problem reduces to finding $g_i \in (v_i, 1)$ satisfying:
		\[
		\epsilon^\star = \frac{m g_i (w_i + g_i - v_i)}{c_i (g_i - v_i)(g_i - \alpha)}
		\]}
	
	{	Rearranging terms yields the quadratic equation:
		\[
		f(x) = (\epsilon^\star c_i - m)x^2 - [\epsilon^\star c_i (\alpha + v_i) + m(w_i - v_i)]x + \epsilon^\star c_i \alpha v_i = 0.
		\]}
	
	{We first establish that $f(v_i) < 0$. Evaluating at $x = v_i$, we have
		\[
		\begin{aligned}
			f(v_i) &= (\epsilon^\star c_i - m)v_i^2 \hspace{-0.1cm}-\hspace{-0.1cm} [\epsilon^\star c_i (\alpha + v_i)\hspace{-0.05cm} + \hspace{-0.05cm} m(w_i - v_i)]v_i \hspace{-0.05cm}+\hspace{-0.05cm} \epsilon^\star c_i \alpha v_i \\
			&= \epsilon^\star c_i v_i^2 \hspace{-0.05cm}-\hspace{-0.05cm} m v_i^2 \hspace{-0.1cm}- \hspace{-0.1cm}\epsilon^\star c_i (\alpha + v_i)v_i \hspace{-0.05cm}-\hspace{-0.05cm} m(w_i - v_i)v_i \hspace{-0.05cm}+ \hspace{-0.05cm}\epsilon^\star c_i \alpha v_i \\
			&= -m w_i v_i < 0.
		\end{aligned}
		\]
		Next, evaluating at $x = 1$, we have
		\[
		\begin{aligned}
			f(1) &= (\epsilon^\star c_i - m) - [\epsilon^\star c_i (\alpha + v_i) + m(w_i - v_i)] + \epsilon^\star c_i \alpha v_i \\
			&= \epsilon^\star c_i (1 - \alpha)(1 - v_i) - m(w_i + 1 - v_i)
		\end{aligned}
		\]}
	
	{Assume $m(w_i + 1 - v_i) < \epsilon^\star c_i (1 - \alpha)(1 - v_i)$. Then, $f(1) > 0$. Since $f(v_i) < 0$, $f(1) > 0$, and $f(x)$ is continuous, the Intermediate Value Theorem guarantees at least one real root in $(v_i, 1)$. Uniqueness follows from analyzing the quadratic's behavior: regardless of the sign of $(\epsilon^\star c_i - m)$, the opposite signs at the endpoints ensure exactly one root in the interval.
		This completes the proof.}
\end{proof}

	\
	
	{\begin{remark}
			A notable distinction in Theorem \ref{CDP2} is the absence of the observer gain $\overline{L}$ in the privacy condition, contrasting with its presence in Theorem \ref{CDP1} through the term $\|L\|_1$. This fundamental difference arises from the intrinsic structure of the reduced-order observer. The observer is designed to estimate only the unmeasurable portion of the states, while the outputs $y(k)$ itself constitutes the directly accessible information that requires privacy protection.  The estimation process of the internal states, governed by $\overline{L}$, does not directly expose sensitive output information and therefore does not influence the privacy guarantee. This reveals an inherent advantage of reduced-order observers in privacy-preserving control systems: the mechanisms for state estimation and privacy-preserving can be decoupled.
	\end{remark}}
	
	{	\begin{remark}
			Compared with the full-order observer, the reduced-order observer significantly decreases the computational burden. 
			Since only the reduced state $\hat{\overline{x}}_{1 i}(k)\in\mathbb{R}^{n-q}$ is estimated, the observer update requires 
			$O((n-q)^{2})$ operations, instead of the $O(n^{2})$ operations in the full-order case. 
			The control computation $O(d_i n)$ and the $O(n)$ Laplace noise generation remain unchanged. 
			Therefore, the per-agent complexity decreases from $O(n^{2})$ (full-order) to 
			\[
			O((n-q)^{2}) + O(d_i n) + O(n),
			\]
			achieving a strictly lower computational cost. 
			This reduction becomes especially significant when $q$ is large, offering improved scalability for high-dimensional systems.
	\end{remark}}
	
	{	\begin{remark}
			Our framework holds potential for extension to directed graph scenarios.  Although the detailed design of such algorithms is beyond the scope of this paper, one can follow the core ideas presented herein-particularly regarding noise adding and observer design-to develop  differentially private consensus protocols for directed networks. In the case of directed graphs, it is crucial to note that the system converges to a weighted average consensus rather than  a dynamic  average consensus. This stems from the spectral properties of the asymmetric Laplacian matrix, which cause the system to converge to a weighted average. The steady-state weights for this average are determined by the corresponding left eigenvector \cite{erfen}.
		\end{remark}
	}	
	{\begin{remark}
			It is worth mentioning that random delays and packet losses make the closed-loop networked system history-dependent and generally {non-Markovian}, which prevents a direct extension of our current convergence (consensus) analysis to such settings \cite{PL}. Therefore, establishing differentially private consensus under non-ideal communications with delays and packet losses calls for new privacy-preserving algorithms and corresponding analysis tools. Neither communication delays nor packet losses weaken the privacy protection, since the differential privacy noise is injected locally before transmission. This is consistent with the fact that differential privacy is resilient to postprocessing (see Theorem 1 in \cite{LeNy2014}). We leave this important direction for future research.
	\end{remark}}
	\section{Simulation}\label{S5}
	{	In this section, we consider the MAS \eqref{mainagent} with the communication topology depicted in Fig. 1.} The system matrices are given by
	$$
	\begin{gathered}
		A=\left[\begin{array}{cc}
			1.2 & 0 \\
			0 & 0.5
		\end{array}\right], B=\left[\begin{array}{cc}
			1 & 0 \\
			0 & 1
		\end{array}\right], C=\left[\begin{array}{ll}
			1 & 0 
		\end{array}\right],
	\end{gathered}
	$$
	The simulation results are presented separately based on the full-order observer and the reduced-order observer.
	\begin{figure}[!t]
		\centerline{\includegraphics[width=88mm]{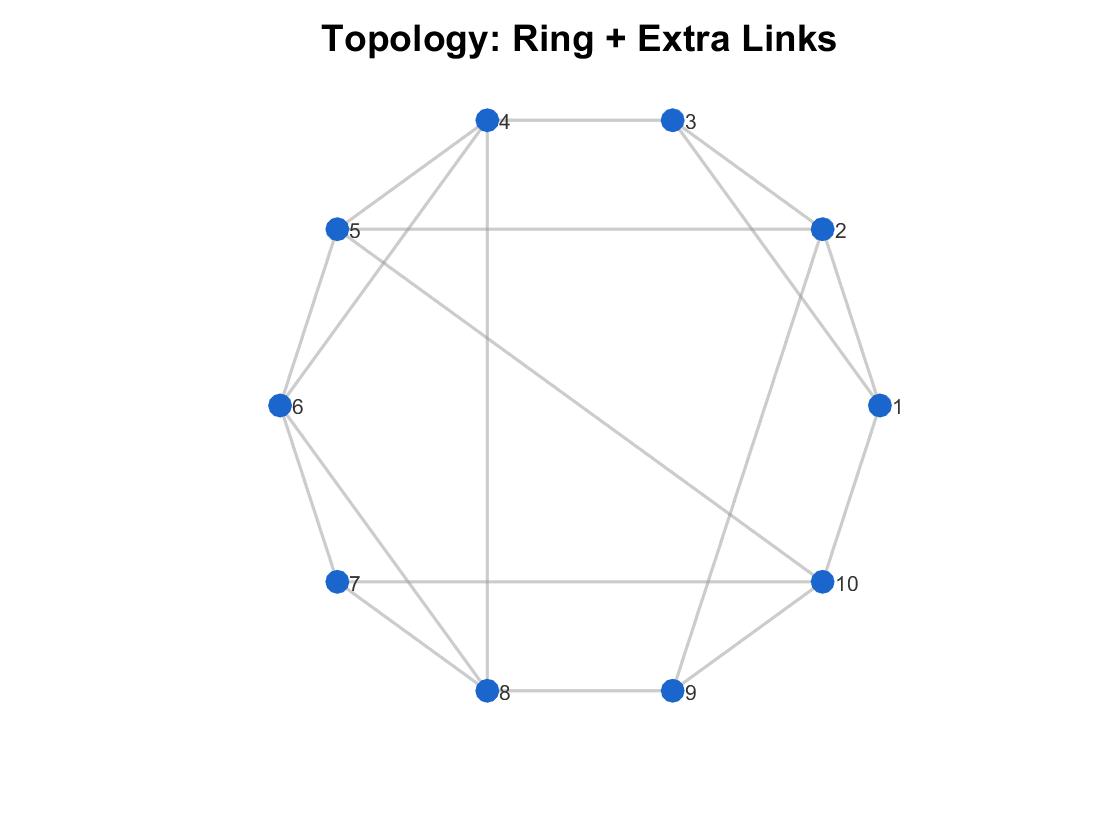}}
		\caption{ {The communication topology of MAS (1).}}
		\label{To}
	\end{figure}

	Example 1 (Full-order observer):
	Selecting
	$$
	\begin{gathered}
		L=\left[\begin{array}{ll}
			0.5 \\
			0.45
		\end{array}\right],
		K=\left[\begin{array}{ll}
			0.18 & 0 \\
			0 & 0
		\end{array}\right]
	\end{gathered}
	$$
	such that $\rho(A-L C)<1$ and $\rho(I_{N-1} \otimes A-\Lambda \otimes B K)< 1$.

	To avoid resonance phenomena, the observer gain parameters $g_i$ and $c_i$ are initialized as random values uniformly distributed within the intervals $[0.9, 0.95]$ and $[1.2, 1.24]$, respectively, where $i = 1, 2, \ldots, 10$.
	
	Denote $\theta_i^a(k)$ as the first component of $\theta_i(k)$. For adjacent $y(k)$ and $y^\prime(k)$, denote $\theta_i^a(k)$ and $(\theta_i^a(k))^\prime$ as the message that Agent $i$ sent to other agents at time $k$, respectively. It can be computed that  $max(l_i)=0.86<g_i$.  Let $h(k)=m \alpha^k$ with $m=0.5$, $\alpha=0.5<min(l_i)=0.68$, $y_1^\prime(k)=y_1(k)+m\alpha^k$, and $y_i^\prime(k)=y_i(k)$ for $i=2,3,\cdots,10$.
	
	Therefore, all conditions of Corollary \ref{CDP1} are satisfied. The mean square consensus, almost sure consensus, and the preservation of $\epsilon$-differential privacy are simultaneously achieved with $\epsilon=23.4$. As shown in Fig 2  and Fig 3, the mean square and almost sure consensus are achieved.  We carried out $1000$ experiments and the histogram of $\theta_1^a(k)$ and $(\theta_1^a(k))^\prime$ at $k = 2$ is given in Fig 4.
	\begin{figure}[!t]
		\centerline{\includegraphics[width=88mm]{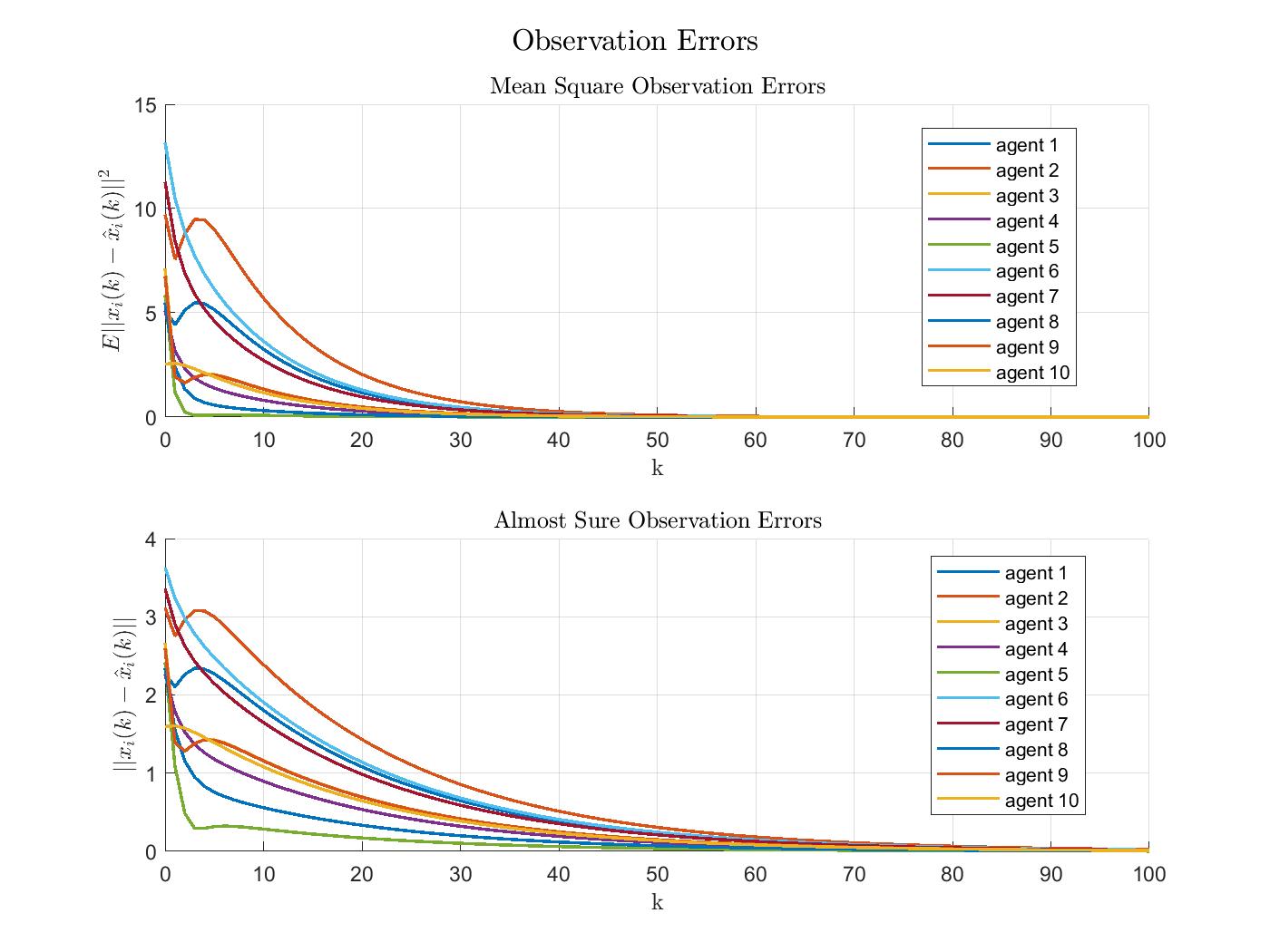}}
		\caption{The trajectory of observation errors (Full-order Observer).}
		\label{1error1}
	\end{figure}
	
	\begin{figure}[!t]
		\centerline{\includegraphics[width=88mm]{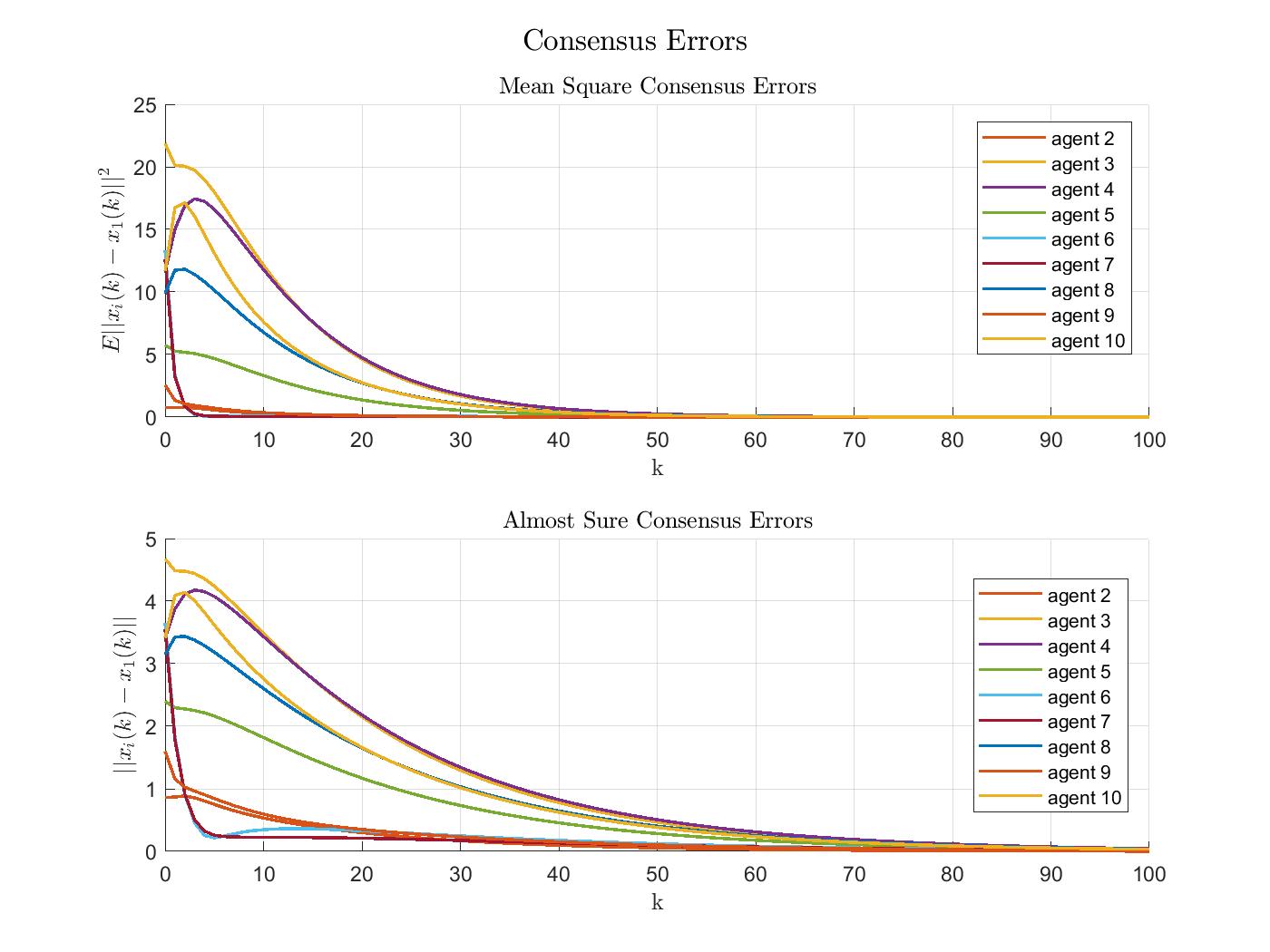}}
		\caption{The trajectory of consensus errors (Full-order Observer).}
		\label{2error2}
	\end{figure}

	\begin{figure}[!t]
		\centerline{\includegraphics[width=88mm]{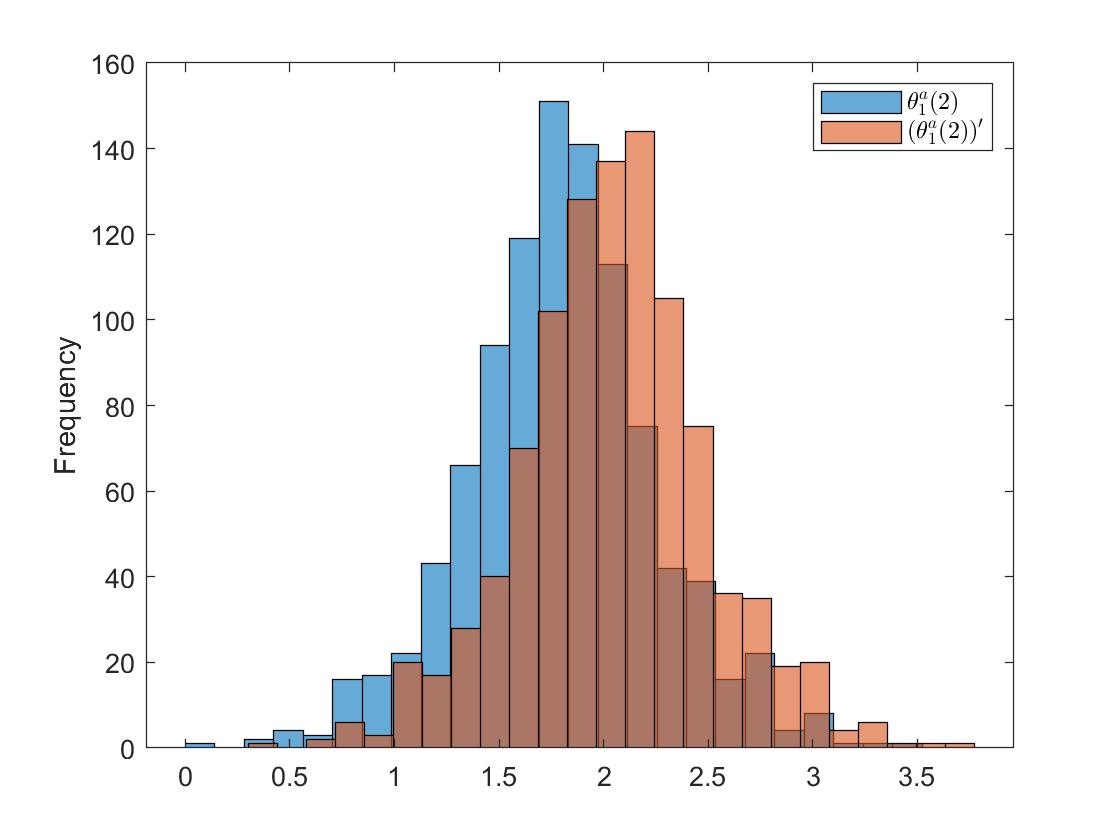}}
		\caption{The histograms of $\theta_1^a(k)$ and $(\theta_1^a(k))^\prime$ at $k = 2$ (Full-order Observer).}
		\label{2consensusb}
	\end{figure}
	
	\
	
	Example 2 (Reduced-order observer):
	Selecting
	$$
	\begin{gathered}
		P=\left[\begin{array}{ll}
			1 & 1 \\
			1 & 0
		\end{array}\right],
	\end{gathered}
	$$
	we have
	$$
	\begin{gathered}
		\overline{A}=PAP^{-1}=\left[\begin{array}{cc}
			0.5 & 1 \\
			0 & 1.5
		\end{array}\right], \\
		\overline{B}=PB=\left[\begin{array}{cc}
			1 & 1 \\
			1 & 0
		\end{array}\right], \overline{C}=CP^{-1}=\left[\begin{array}{ll}
			0 & 1
		\end{array}\right].
	\end{gathered}
	$$
	Select $\overline{L}=0.03$ such that $\rho(\overline{A}_{11}-\overline{L}\overline{A}_{21})=0.5<1$.
	
	Select
	$$
	\begin{gathered}
		K=\left[\begin{array}{ll}
			0.05 & 0 \\
			0 & -0.02
		\end{array}\right]
	\end{gathered}
	$$
	such that $\rho(I_{N-1} \otimes A-\Lambda \otimes B K)< 1$.
	
	To avoid resonance phenomena, the observer gain parameters $g_i$ and $c_i$ are initialized as random values uniformly distributed within the intervals $[0.9, 0.95]$ and $[0.5, 0.54]$, respectively, where $i = 1, 2, \ldots, N$.
	Therefore, all conditions of Corollary \ref{coro6} are satisfied.
	As shown in Fig 5 and Fig 6, the mean square and almost sure consensus are achieved.
	
	Denote $\theta_i^a(k)$ as the first component of $\theta_i(k)$. For adjacent $y(k)$ and $y^\prime(k)$, denote $\theta_i^a(k)$ and $(\theta_i^a(k))^\prime$ as the message that Agent $i$ sent to other agents at time $k$, respectively. It can be  computed that  $max(v_{i})=0.1$.  Let $h(k)=m \alpha^k$ with  $m=0.5$, $\alpha=0.5$, $y_1^\prime(k)=y_1(k)+m\alpha^k$, and $y_i^\prime(k)=y_i(k)$ for $i=2,3,\cdots,10$. It can be seen that $\max(\alpha, v_i)<g_{i}<1$. From Corollary \ref{con5}, the $\epsilon$-differential privacy is preserved with $\epsilon=4.4$. We carried out $1000$ experiments and the histogram of $\theta_1^a(k)$ and $(\theta_1^a(k))^\prime$ at $k = 4$ is given in Fig 7.
	\begin{figure}[!t]\label{ReObservation}
		\centerline{\includegraphics[width=88mm]{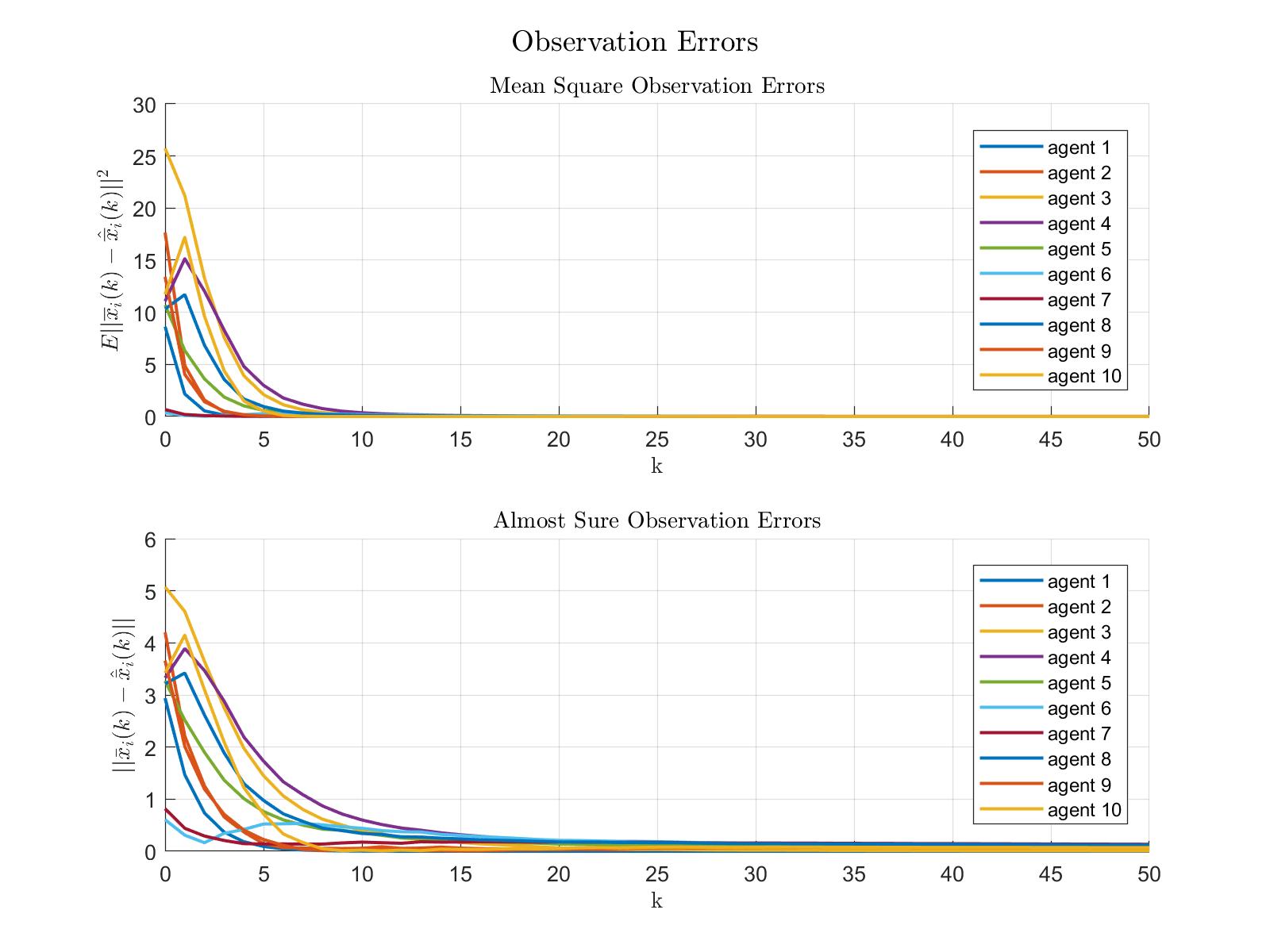}}
		\caption{The trajectory of observation errors (Reduced-order Observer).}
		\label{remserror}
	\end{figure}
	
	\begin{figure}[!t]\label{ReConsensus}
		\centerline{\includegraphics[width=88mm]{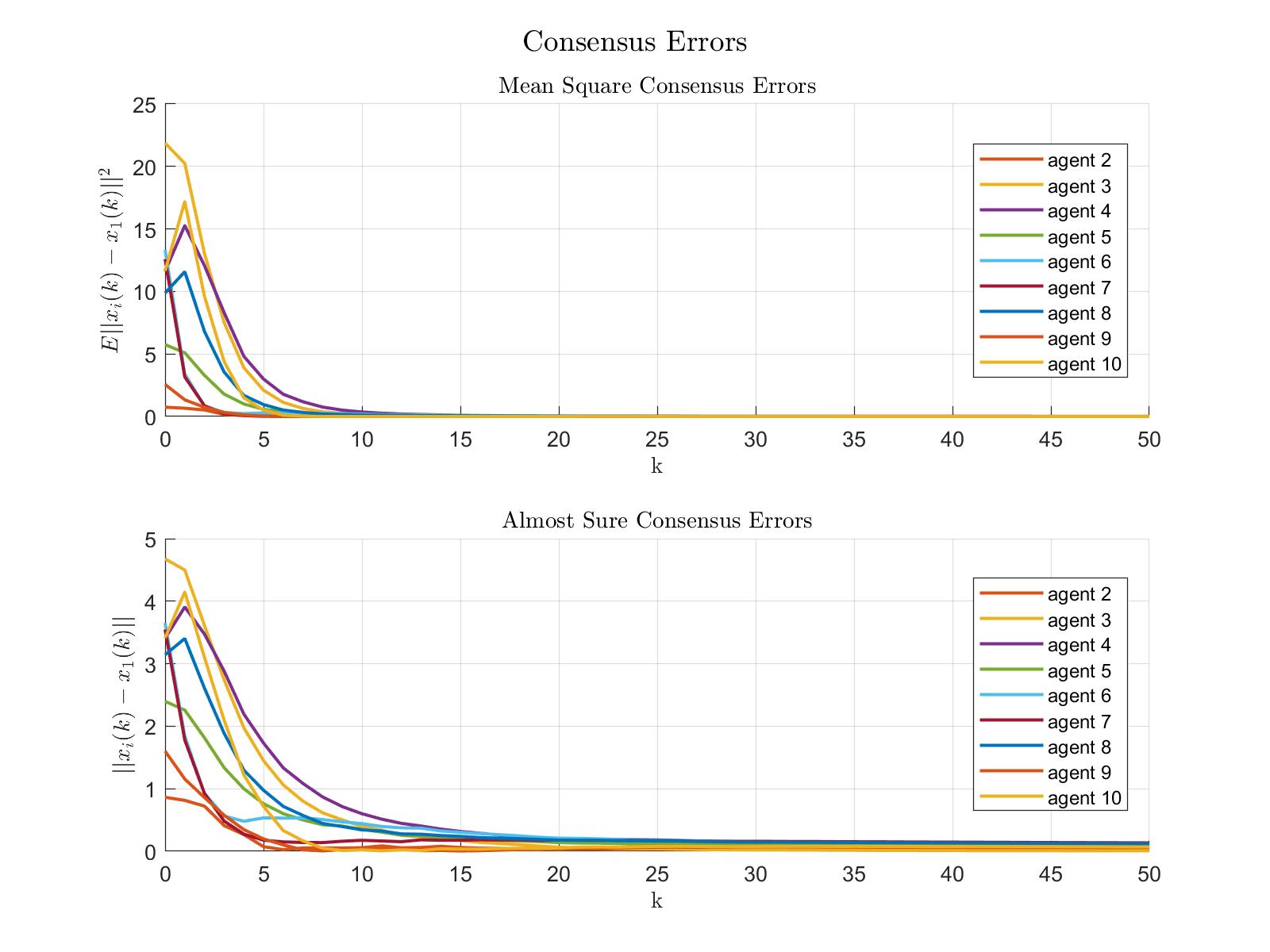}}
		\caption{The trajectory of consensus errors (Reduced-order Observer).}
		\label{remsconsensus}
	\end{figure}

	\begin{figure}[!t]	\centerline{\includegraphics[width=88mm]{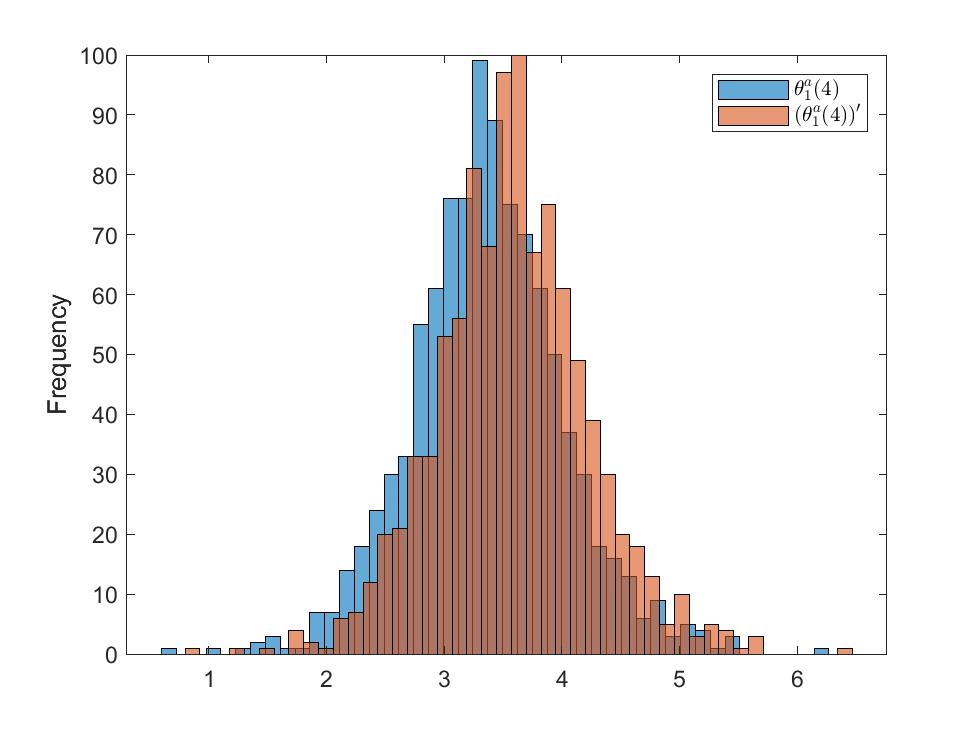}}
		\caption{The histograms of $\theta_1^a(k)$ and $(\theta_1^a(k))^\prime$ at $k = 4$ (Reduced-order Observer).}
		\label{Laplace2}
	\end{figure}
	
	{\begin{remark}
			It is worth noting that the proposed framework admits natural scalability to large-scale networks. As established by our theoretical analysis, such as Corollary 5 and Corollary 10, the validity of the differentially private consensus results does not depend on the number of agents, and larger networks with dozens or even hundreds of agents are expected to exhibit the same privacy-preserving consensus behavior. 
	\end{remark}}
	
	{	\begin{remark}
			This paper adopts differential privacy via noise injection to provide a preservation of output information. It avoids key management and incurs low computational and communication overhead, making it easy to deploy in MASs. Its limitation is the inherent privacy-performance trade-off: injected noise may slow convergence or reduce steady-state accuracy. By contrast, homomorphic encryption can offer stronger, cryptographic (deterministic) privacy but at the cost of heavy computation and communication and key distribution. Affine masking and state-decomposition methods typically preserve accuracy better without adding noise, yet rely on specific structural or invertibility assumptions, which restrict their applicability. 
		\end{remark}
	}
	\section{Concluding remarks}\label{S6}
	This paper investigates differentially private consensus in general linear multi-agent systems using both full-order and reduced-order observers. To protect the potentially vulnerable output information, Laplace noise is introduced into the inter-agent communication process. The sufficient conditions to ensure differentially private consensus based on full-order and reduced-order observers are derived, respectively. The results not only deepen the theoretical understanding of privacy-preserving consensus in MASs but also provide a practical framework for protecting sensitive information in distributed cooperative control.
	
	{	The present study focuses on static communication topologies. Extending these results to switching topologies represents a meaningful direction for future work. In such dynamic settings, protocol stability and convergence would likely be influenced by network connectivity and switching patterns. While differential privacy mechanisms, depending on noise statistics, may exhibit inherent robustness to topological changes, the network structure could still indirectly shape the privacy-performance trade-off. For example, weaker connectivity might prolong convergence, potentially increasing the exposure of intermediate states \cite{ST,ST1}.}

	{	In the future,  implementing differentially private consensus control within a fully distributed architecture presents a promising research avenue. One might explore the use of local rules, such as Metropolis weights, in combination with time-varying noise mechanisms. Such an approach could potentially achieve both mean-square consensus and privacy protection over an infinite time horizon under fixed topologies \cite{Fully}.}
	
	{	Extending the differential privacy consensus framework from existing linear systems to more general nonlinear scenarios may require leveraging a series of advanced analytical tools. For instance, one could explore integrating Lyapunov stability theory with sector-bounded conditions to address complex nonlinear dynamics in the system. The co-design of control and noise injection mechanisms might also be investigated using tools such as linear matrix inequalities \cite{Non}.}
	
	
	Although Laplace and Gaussian noise are widely employed for privacy protection, they often come at the cost of compromising the convergence rate. Future research may focus on developing more efficient noise injection strategies, such as adaptive noise mechanisms\cite{meng2020} and hybrid noise mechanisms\cite{MK2023}, and other  protocols, such as event-triggered protocol \cite{LDLS2024}. The idea of network augmentation is also very interesting\cite{RAK2024}.

	\begin{IEEEbiography}[{\includegraphics[width=1in,height=1.25in,clip,keepaspectratio]{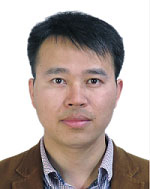}}]
		{Xiaofeng Zong}(Member, IEEE) received the B.S. degree in mathematics from Anqing Normal University, Anhui, China, in 2009, and the Ph.D. degree in probability theory and mathematical statistics from the School of Mathematics and Statistics, Huazhong University of Science and Technology, Hubei, China, in 2014. He was a Postdoctoral Researcher with the Academy of Mathematics and Systems Science (AMSS), Chinese Academy of Sciences (CAS), Beijing, China, from July 2014 to July 2016, and a Visiting Assistant Professor with the Department of Mathematics, Wayne State University, Detroit, MI, USA, from September 2015 to October 2016. Since October 2016, he has been a Professor with the School of Automation, China University of Geosciences, Wuhan, Hubei, China, where now he is an Associate Dean with the School of Automation. His current research interests include stochastic approximations, stochastic systems, delay systems, and multiagent systems.
	\end{IEEEbiography}
	
	\begin{IEEEbiography}[{\includegraphics[width=1in,height=1.25in,clip,keepaspectratio]{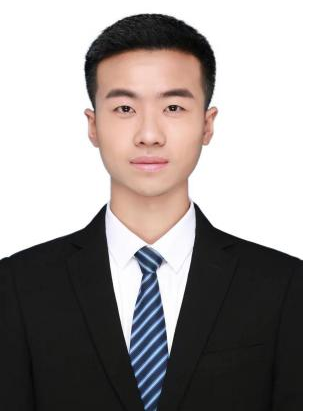}}]
		{Ming-Yu Wang } received the B.E. degree in automation from Central South University in 2022. Now, he  is pursuing a Ph.D.  in automation from China University of Geosciences, Wuhan, Hubei, China. His current research interests include multi-agent systems, differential privacy and stochastic systems.
	\end{IEEEbiography}
	
	\begin{IEEEbiography}[{\includegraphics[width=1in,height=1.25in,clip,keepaspectratio]{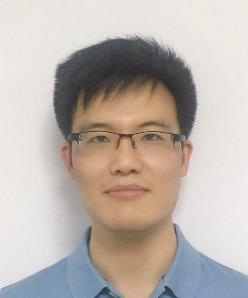}}]
		{Jimin Wang} (Member, IEEE) received the B.S. degree in mathematics from Shandong Normal University, China, in 2012 and the Ph.D. degree from the School of Mathematics, Shandong University, China, in 2018. From May 2017 to May 2018, he was a joint Ph.D. student with the School of Electrical Engineering and Computing, The University of Newcastle, Australia. From July 2018 to December 2020, he was a postdoctoral researcher in the Institute of Systems Science (ISS), Chinese Academy of Sciences (CAS), China. He is currently an associate professor in the School of  Automation and Electrical Engineering, University of Science and Technology Beijing. His current research interests include privacy and security in cyber-physical systems, stochastic systems and networked control systems. He was a recipient of Shandong University's excellent doctoral dissertation.
	\end{IEEEbiography}
	
	\begin{IEEEbiography}[{\includegraphics[width=1in,height=1.25in,clip,keepaspectratio]{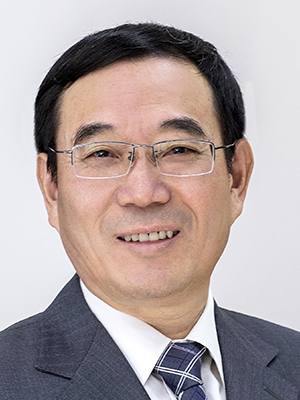}}]
		{Ji-Feng Zhang} (Fellow, IEEE) received the B.S. degree in mathematics from Shandong University, China, in 1985, and the Ph.D. degree from the Institute of Systems Science, Chinese Academy of Sciences (CAS), China, in 1991. Now he is with the School of Automation and Electrical Engineering, Zhongyuan University of Technology; and the State Key Laboratory of Mathematical Sciences, Academy of Mathematics and Systems Science, CAS. His current research interests include system modeling, adaptive control, stochastic systems, and multi-agent systems.
		
		He is an IEEE Fellow, IFAC Fellow, CAA Fellow, SIAM Fellow, member of the European Academy of Sciences and Arts, and Academician of the International Academy for Systems and Cybernetic Sciences. He received the Second Prize of the State Natural Science Award of China in 2010 and 2015, respectively. He was a Vice-President of the Chinese Association of Automation, the Chinese Mathematical Society and the Systems Engineering Society of China. He was a Vice-Chair of the IFAC Technical Board, member of the Board of Governors, IEEE Control Systems Society; Convenor of Systems Science Discipline, Academic Degree Committee of the State Council of China. He served as Editor-in-Chief, Deputy Editor-in-Chief or Associate Editor for more than 10 journals, including Science China Information Sciences, IEEE Transactions on Automatic Control and SIAM Journal on Control and Optimization etc.
	\end{IEEEbiography}

\end{document}